\documentclass[journal]{IEEEtran}
\usepackage[T1]{fontenc}
\usepackage{microtype}
\usepackage{amsmath,amssymb,amsthm}
\usepackage{booktabs}
\usepackage{enumitem}
\usepackage[hidelinks]{hyperref}
\newtheorem{theorem}{Theorem}
\newtheorem{proposition}[theorem]{Proposition}

\newtheorem{lemma}[theorem]{Lemma}
\theoremstyle{definition}

\theoremstyle{remark}

\DeclareMathOperator{\diag}{diag}
\DeclareMathOperator{\rank}{rank}
\DeclareMathOperator{\sech}{sech}
\DeclareMathOperator{\argmin}{arg\,min}
\newcommand{\T}{\mathbb T}
\newcommand{\R}{\mathbb R}
\newcommand{\E}{\mathbb E}
\newcommand{\one}{\mathbf 1}
\newcommand{\ind}{\mathbf 1}
\newcommand{\distT}{d_{\mathbb T}}
\newcommand{\etaN}{\eta_N}
\newcommand{\cM}{\mathcal C_M}
\newcommand{\NM}{\mathcal N_M}
\newcommand{\EM}{\mathfrak E_M}
\newcommand{\eps}{\varepsilon}
\title{Sharp Minimax Limits and Compatibility Spectra for Critical Near-DFT Index-Only Frequency Estimation}
\author{Armon~Rasooli and Mohammad~Sadegh~Narimani%
\thanks{\raggedright ORCID: \mbox{Armon~Rasooli}, \href{https://orcid.org/0009-0002-6583-7090}{0009-0002-6583-7090}; \mbox{Mohammad~Sadegh~Narimani}, \href{https://orcid.org/0009-0004-8648-2238}{0009-0004-8648-2238}.\par}}
\hypersetup{
  pdftitle={Sharp Minimax Limits and Compatibility Spectra for Critical Near-DFT Index-Only Frequency Estimation},
  pdfauthor={Armon Rasooli; Mohammad Sadegh Narimani},
  pdfsubject={Frequency estimation from index-only near-DFT categorical observations},
  pdfkeywords={frequency estimation, channelizer design, categorical observations, asymptotic minimax risk, unitary transforms}
}

\begin{document}
\maketitle
\pagestyle{empty}
\thispagestyle{empty}

\begin{abstract}
An \(M\)-channel discrete Fourier transform (DFT) channelizer routes an
on-grid sinusoid to a single output. When each frame reports only one
energy-proportional channel index, however, signed sub-bin frequency
estimation is nonregular: dark-channel probability is quadratic in the
offset, whereas orientation enters cubically. We study \(N\) independent
labeled reports under a known uniform-replacement probability
\(\varepsilon_N\) and one deterministic unitary shared by all frequencies
and frames, constrained to routing defect \(\tau/N\). If
\(\sqrt N\varepsilon_N\to\lambda<\infty\), we establish an attained
global-in-frequency minimax limit at the critical scales \(N^{-1/4}\) for
frequency offset, \(N^{-1/2}\) for unitary perturbation and replacement, and
\(N^{-1}\) for routing defect. The effective unitary tangent is a symmetric
complete-graph edge field modulo one centering nuisance, and its first
variation is a radius-dependent weighted divergence. For \(\lambda>0\),
evaluation at \(k\) distinct normalized offset magnitudes yields an exact
Fourier--Cauchy nullity spectrum: \(\lfloor(M-1)/2\rfloor\) evaluations are
necessary and sufficient, for every choice of distinct magnitudes, to
certify neutrality at all magnitudes. The terminal nullspace has a
greatest-common-divisor dimension formula and positive-definite aggregate
curvature. Consequently, exact DFT routing is uniquely minimax within the
complete critical tangent class for \(M=3\), whereas every positive critical
defect budget strictly improves the minimax constant for \(M\ge4\). The
analysis also yields a smallest-prime curvature-visibility law and, for
\(M\ge5\), discontinuous compatibility geometry but continuous minimax value
at the zero critical replacement floor \(\lambda=0\).
\end{abstract}


\section{Introduction}

Frequency refinement from DFT coefficients is a basic signal-processing
problem \cite{RifeBoorstyn1974,Kay1989,Quinn1994,AboutaniosMulgrew2005,
Macleod1998}. Quantization changes its
information geometry and can destroy conclusions inherited from full complex
samples \cite{HostMadsenHandel2000}. Here the front end is more
severe: each independent frame produces only one labeled channel index,
sampled according to normalized transform energy. This idealizes randomized
channel reporting or an energy-proportional event sampler. Observing the full
transform vector or all channel powers is a different, more informative
experiment.

More explicitly, a frame transmits one symbol \(J\in\{0,\ldots,M-1\}\),
not the vector \(Ux_\theta\). If \(e_J\) is its one-hot representation, then
\begin{equation}
 \E[e_J\mid\theta,U]=p_\theta^U.\label{eq:onehot}
\end{equation}
The empirical histogram therefore estimates the transform-energy law across
frames while the interface remains one \(M\)-ary label per frame. This is why
the categorical experiment, rather than direct transform readout, is the
object of study.

Exact DFT routing and local Fourier expansions are classical. Nonregular
optimal design and minimax-testing arguments are established in general form
\cite{LinMartinYang2019,Tsybakov2009}, while unitary modulus-square tangents and Fourier
defect arithmetic have independent theories
\cite{TadejZyczkowski2008,Banica2013}. Graph divergence and cycle spaces are
classical \cite{JiangLimYaoYe2011}; Cauchy nonsingularity supplies a rank tool
used below \cite{Schechter1959}; and prime/composite Fourier support has its
own uncertainty theory \cite{DonohoStark1989,Tao2005,Meshulam2006}. A distinct
three/four transition also occurs in unistochastic geometry at the flat
matrix \cite{BengtssonEtAl2005}. These works provide constituent analytical
tools, but they do not characterize the attained minimax constant or the
all-radius compatibility filtration induced by a shared near-routing unitary
under labeled categorical observations.

Algebraically, the exact DFT categorical law is also a finite-register
quantum-phase-estimation readout law. General phase-estimation work optimizes
circuits, probe tapers, signal polynomials, or Bayesian risk
\cite{VanDamEtAl2007,BerryEtAl2009,ChapeauBlondeauBelin2020,
SmithEtAl2025,PatelEtAl2026,JiaLiuDong2026}. Those results
neither impose the present near-routing defect ball nor optimize one shared
final unitary across grid nodes under worst-case fixed-\(M\), repeated-frame
risk. Conversely, our theorems do not optimize probe states, coherent query
complexity, or general quantum measurements.

This work makes two principal contributions. First, it establishes a sharp
global-in-frequency asymptotic minimax law over the \(O(N^{-1})\)
channelization-defect class. The lower bound and a single finite-data
estimator attain the same constant uniformly over frequency, including the
singular zero-critical-replacement regime \(\lambda=0\). Second, for
\(\lambda>0\), it characterizes the complete first-order-neutral subspace of
shared unitary perturbations. The stacked operator evaluated at \(k\)
distinct normalized offset magnitudes has an exact Fourier--Cauchy nullity
spectrum; \(L=\lfloor(M-1)/2\rfloor\) evaluations are necessary and
sufficient, for every choice of distinct magnitudes, to certify neutrality at
all magnitudes, and the terminal kernel and curvature are explicit. These
results yield three structural consequences: an \(M=3\) objective-specific
critical rigidity theorem with a strict \(M\ge4\) improvement alternative; a
smallest-prime curvature-visibility law; and, for \(M\ge5\), a discontinuous
compatibility kernel but continuous optimized minimax value at \(\lambda=0\).
The novelty lies in this joint statistical-design characterization; the
Fourier expansion, graph divergence, Cauchy determinant, and finite-group
arithmetic used in its derivation are classical. Section~II specifies the
experiment, Sections~III--VI develop the canonical geometry and the two
principal results, and Section~VII derives their design consequences.
Detailed uniform remainder bounds and proof details are provided in the
Supplementary Material.

\begin{table*}[t]
\caption{Assumptions and implications of the principal results.}
\label{tab:scope}
\centering
\footnotesize
\begin{tabular}{@{}p{0.10\textwidth}p{0.40\textwidth}p{0.42\textwidth}@{}}
\toprule
Result & Assumptions & Established conclusion \\
\midrule
Minimax limit & Fixed \(M\ge3\); finite \(\tau\) and \(\lambda\); known uniform
replacement with \(\sqrt N\varepsilon_N\to\lambda\); one deterministic
unitary shared by all frequencies and frames; labeled IID categorical reports;
squared circular loss. & The normalized global minimax risk converges to the
attained constant \(\mathcal C_M(\tau,\lambda)\), and a single measurable
estimator attains this limit uniformly over frequency. \\
Compatibility spectrum & Positive critical replacement floor \(\lambda>0\); \(k\)
distinct positive normalized offset magnitudes; complete tangent graph
\(K_M\), with every unordered pair coordinate admissible;
quotient by the common edge translation. & The \(k\)-evaluation neutral
space has the exact Fourier--Cauchy nullity spectrum in
Theorem~\ref{thm:G2}; \(\lfloor(M-1)/2\rfloor\) evaluations are necessary and
sufficient for all-magnitude certification, with an explicit terminal kernel
and curvature. \\
Consequences & The regimes stated in Theorem~\ref{thm:M3global} and
Propositions~\ref{prop:localization}--\ref{prop:decouple}. & Exact DFT routing
is critically rigid for \(M=3\) and strictly improvable for \(M\ge4\);
persistent modes obey the smallest-prime curvature-visibility law; for
\(M\ge5\), the compatibility kernel jumps at \(\lambda=0\) while the
optimized minimax value remains continuous. \\
\bottomrule
\end{tabular}
\end{table*}

The results concern fixed \(M\), known uniform replacement, a single
deterministic shared unitary, and labeled categorical observations. They do
not address growing \(M\), adaptive transforms, unknown or nonuniform
corruption, or full-vector or power observations. The compatibility spectrum
is an exact-rank statement; it does not imply uniform conditioning as the
evaluated offset magnitudes coalesce or as \(\lambda\downarrow0\). Global
optimality refers to the complete critical shrinking-defect class, not to
arbitrary finite-distance unitaries.

\section{Experiment and boundary geometry}

Fix \(M\ge3\), write \(\T=\R/(2\pi\mathbb Z)\), and set
\[
 x_\theta[n]=M^{-1/2}e^{in\theta},\quad n=0,\ldots,M-1.
\]
For \(U\in U(M)\), one frame reports
\[
 J\mid(\theta,U)\sim\operatorname{Categorical}(p_\theta^U),\qquad
 p_\theta^U(j)=|(Ux_\theta)_j|^2.
\]
Independently replace \(J\) by a uniform label with known probability
\(\eps_N\); thus
\(\widetilde p_{\theta,N}^{U}=(1-\eps_N)p_\theta^U+\eps_N\one/M\).
The \(N\) reports are IID and labeled. Risk uses squared circular distance
\(\distT^2\).

Let \(\theta_q=2\pi q/M\), \(X=[x_{\theta_0},\ldots,x_{\theta_{M-1}}]\),
and \(U_0=X^\dagger\). Orthogonality of finite Fourier characters gives
\(U_0x_{\theta_q}=e_q\). Every exact router is \(DPU_0\), with diagonal
unitary \(D\) and permutation \(P\).

For \(\theta=\theta_q+h\), the exact DFT probabilities are
\[
 p_{qj}(h)=
 \left[\frac{\sin(Mh/2)}{M\sin\{\pi(q-j)/M+h/2\}}\right]^2,\quad j\ne q,
\]
and \(p_{qq}(h)=[\sin(Mh/2)/(M\sin(h/2))]^2\). Put
\(\alpha_{qj}=\pi(q-j)/M\). Taylor expansion yields
\begin{align}
 p_{qj}(h)&=a_{qj}h^2+b_{qj}h^3+O(h^4),\label{eq:abexp}\\
 a_{qj}&=\frac1{4\sin^2\alpha_{qj}},\qquad
 b_{qj}=-\frac{\cos\alpha_{qj}}{4\sin^3\alpha_{qj}},\label{eq:ab}\\
 p_{qq}(h)&=1-\frac{M^2-1}{12}h^2+O(h^4).\label{eq:routed}
\end{align}
Thus \(a_{qj}=a_{jq}\), \(b_{qj}=-b_{jq}\), and
\(\sum_{j\ne q}a_{qj}=(M^2-1)/12\). Direct evaluation at \(M=2\) gives
\(p_{q,q+1}(h)=\sin^2(h/2)=p_{q,q+1}(-h)\), an exact collision. For
\(M\ge3\), at least the channel \(j=q+1\) has \(b_{qj}\ne0\). Its first
circular label moment is
\[
 m_1(\theta)=\frac{(M-1)e^{i\theta}+e^{-i(M-1)\theta}}M.
\]
If \(m_1(\theta)=m_1(\phi)\) with \(z=e^{i\theta}\ne w=e^{i\phi}\), division
by \(z-w\) makes a sum of \(M-1\) unit phasors have modulus \(M-1\).
Equality in the triangle inequality forces \(w/z=1\) when \(M\ge3\), a
contradiction. Hence the labeled exact-DFT law is globally injective.

The squared Hellinger distance between the laws at \(\theta_q\pm h\) is
\(\Theta(h^4)\): the probability difference is \(\Theta(h^3)\) while the
corresponding dark mass is \(\Theta(h^2)\). Consequently \(Nh^4\asymp1\)
and \(h\asymp N^{-1/4}\). A uniform floor competes with dark mass precisely
when \(\eps_N\asymp h^2\asymp N^{-1/2}\).

\section{Critical unitary design reduction}

Define the channelization defect
\[
 \delta_{\rm ch}(U)=
 \min_{\pi\in S_M}\max_q\{1-| (Ux_{\theta_q})_{\pi(q)}|^2\},
\]
and \(\mathcal U_N(\tau)=\{U:\delta_{\rm ch}(U)\le\tau/N\}\). If
\(U_N\in\mathcal U_N(\tau)\), phase/permutation alignment gives
\[
 U_N=e^{N^{-1/2}K_N}U_0,\qquad K_N^\dagger=-K_N,\quad \sup_N\|K_N\|<\infty.
\]
Indeed the aligned off-diagonal squared mass is \(O(N^{-1})\), so the
principal logarithm is \(O(N^{-1/2})\).

Let
\[
 v_{qj}=\left.\partial_h(U_0x_{\theta_q+h})_j\right|_{0}
 =\frac{e^{-i\alpha_{qj}}}{2\sin\alpha_{qj}}.
\]
Only the real phase-aligned coordinates
\[
 c_{qj}=\frac{\Re(v_{qj}^*K_{jq})}{a_{qj}}=c_{jq}
\]
enter the signed critical law. Cauchy--Schwarz shows that the least-resource
realization is \(K(c)_{jq}=c_{qj}v_{qj}\) for \(j\ne q\), with
\(K(c)_{qq}=0\). Thus the effective variables form
a symmetric edge field on \(K_M\); unused quadratures consume defect without
improving the limit.

Write
\[
 \theta=\theta_q+Sr\etaN+\kappa\etaN^2,\qquad
 \etaN=N^{-1/4},\quad S\in\{-1,1\}.
\]
The signed dark-cell coefficient depends on \(c_{qj}+\kappa\). Hence
\(c\mapsto c+\beta\one\), \(\kappa\mapsto\kappa-\beta\) is a nuisance
action, and
\[
 \EM=\R^{E(K_M)}/\operatorname{span}\{\one\}.
\]
Its routing quotient norm is
\begin{equation}
 \|[c]\|_{\rm rt}^2=
 \min_{\beta\in\R}\max_q\sum_{j\ne q}a_{qj}(c_{qj}+\beta)^2.\label{eq:rtnorm}
\end{equation}
The minimizing gauge is unique. Indeed, writing the objective in
\eqref{eq:rtnorm} as \(f_c(\beta)\) and setting
\(A_\Sigma=\sum_{j\ne q}a_{qj}=(M^2-1)/12\), the function
\(f_c(\beta)-A_\Sigma\beta^2\) is convex. Its minimizer
\(\beta_\ast\) satisfies
\[
 f_c(\beta)\ge f_c(\beta_\ast)+A_\Sigma|\beta-\beta_\ast|^2.
\]
The quotient norm is zero only for the constant class. Uniformly on bounded
fields,
\(N\delta_{\rm ch}(e^{N^{-1/2}K(c)}U_0)=
\max_q\sum_{j\ne q}a_{qj}c_{qj}^2+o(1)\); boundary fields are shrunk by
\(1-o(1)\) to meet the finite budget exactly.

\section{Canonical decision functional}

Assume \(\lambda_N=\sqrt N\eps_N\to\lambda\in[0,\infty)\). For \(j\ne q\),
the exact probability has the compact-local expansion
\begin{align}
 \widetilde p_{\theta,N}^{U_N}(j)
 &=\etaN^2 A_{N,qj}(r)+S\etaN^3d_{qj}(r,c,\kappa)+o(\etaN^3),\label{eq:cell}\\
 A_{N,qj}(r)&=(1-\eps_N)a_{qj}r^2+\lambda_N/M,\nonumber\\
 d_{qj}(r,c,\kappa)&=r\{b_{qj}r^2+2a_{qj}(c_{qj}+\kappa)\}.\nonumber
\end{align}
Centered dark counts, scaled by \(N^{-1/4}\), therefore have a
finite-dimensional Gaussian limit with diagonal covariance
\(A_{qj}(r,\lambda)=a_{qj}r^2+\lambda/M\) and signed mean \(Sd_q\).
This assertion concerns the sign-sensitive component on the stated local
sets; it is not a claim of unrestricted equivalence of the full experiment.

Profiling the centre nuisance gives
\begin{equation}
 \mu_{q,\lambda}^2(r,[c])=
 \min_{\kappa\in\R}\sum_{j\ne q}
 \frac{r^2\{b_{qj}r^2+2a_{qj}(c_{qj}+\kappa)\}^2}
 {a_{qj}r^2+\lambda/M}.\label{eq:mu}
\end{equation}
At the singular endpoint, define \(\mu^2\) by the radial limit
\(r\downarrow0\) with \(\lambda=0\); at \(r=0\) it enters the risk only
through the vanishing prefactor \(r^2\). The projected sign statistic is
\(Y=\mu S+Z\), \(Z\sim N(0,1)\). With
\[
 \psi(s)=\E[\sech^2(s+\sqrt{s}Z)],
\]
the minimax estimator of \(rS\) is \(r\tanh(\sqrt{s}Y)\), and its risk is
\(r^2\psi(s)\). The function \(\psi\) is continuous and strictly decreasing;
its derivative is the negative posterior-variance square
\cite{GuoWuShamaiVerdu2011}.

The same profile has a useful quotient-distance form.  Put
\begin{equation}
 w_{qj}(r,\lambda)=\frac{a_{qj}^2}{a_{qj}r^2+\lambda/M},\qquad
 u_{qj}(r)=\frac{b_{qj}r^2}{2a_{qj}},\label{eq:wu}
\end{equation}
and
\(
 \|[z]\|_{q,r,\lambda}^2=
 \min_{\alpha\in\R}\sum_{j\ne q}w_{qj}(z_{qj}+\alpha)^2.
\)
Completing the square, with no approximation, gives
\begin{equation}
 \mu_{q,\lambda}^2(r,[c])
 =4r^2\|[c_q]+[u_q(r)]\|_{q,r,\lambda}^2.\label{eq:qdist}
\end{equation}
Thus the statistical design problem is a simultaneous weighted approximation
of the radius curve \(-[u_q(r)]\) by the stars of one shared edge field.

Define
\begin{align}
 R_{q,r,\lambda}([c])&=r^2\psi(\mu_{q,\lambda}^2(r,[c])),\nonumber\\
 F_{M,\lambda}([c])&=\max_q\sup_{r\ge0}R_{q,r,\lambda}([c]),\label{eq:F}\\
 \cM(\tau,\lambda)&=
 \min_{\|[c]\|_{\rm rt}^2\le\tau}F_{M,\lambda}([c]).\label{eq:C}
\end{align}
Uniformly on a resource ball, \(R\to0\) as \(r\downarrow0\) and as
\(r\to\infty\): in the latter limit
\(\mu^2=r^4\sum b_{qj}^2/a_{qj}+O(r^2)\). Thus the radius supremum is
effectively compact, \(F\) is continuous, and the finite-dimensional minimum
in \eqref{eq:C} is attained.

\section{Positive-Critical-Replacement Compatibility Spectrum}

At \([c]=0\), symmetry makes the optimal centre \(\kappa=0\). Differentiating
\eqref{eq:mu} gives the first-variation operator
\begin{equation}
 [\mathcal D_{M,r,\lambda}g]_q=
 4r^4\sum_{j\ne q}
 \frac{a_{qj}b_{qj}}{a_{qj}r^2+\lambda/M}g_{qj}.\label{eq:D}
\end{equation}
Each edge contributes with opposite sign at its endpoints, so
\(\sum_q[\mathcal Dg]_q=0\). More precisely,
\(\mathcal D_{M,r,\lambda}=4r^4B\Omega_{r,\lambda}\), where \(B\) is an
oriented incidence matrix of the nonzero-weight graph and
\(\Omega_{r,\lambda}\) is an invertible diagonal edge weighting on that
graph. The graph is \(K_M\) for odd \(M\), and \(K_M\) minus the diameter
matching for even \(M\); both are connected for \(M\ge3\). Therefore
\begin{equation}
 \operatorname{im}\mathcal D_{M,r,\lambda}=\one^\perp,
 \qquad \rank\mathcal D_{M,r,\lambda}=M-1.\label{eq:Drank}
\end{equation}
Consequently, for every fixed \(r>0\) and every \(\lambda\ge0\), the
single-evaluation quotient kernel has dimension \(M(M-3)/2\). The smaller
positive-floor persistent space arises only by intersecting the neutrality
conditions across distinct offset magnitudes.

Let \(L=\lfloor(M-1)/2\rfloor\). For \(\lambda>0\) and pairwise distinct
normalized offset magnitudes \(r_1,\ldots,r_k>0\), define
\[
 \NM^{(k)}=\bigcap_{s=1}^k
 \ker\overline{\mathcal D}_{M,r_s,\lambda},\qquad
 \NM(\lambda)=\bigcap_{r>0}
 \ker\overline{\mathcal D}_{M,r,\lambda}.
\]
Here \(\overline{\mathcal D}\) is the operator induced on \(\EM\);
paired antisymmetry makes \(\mathcal D\one=0\), so it is well defined.
For \(0\le\ell<M\), set
\begin{equation}
 \iota_\ell=\#\{1\le d\le L:M\mid \ell d\}
 =\left\lfloor\frac{L\gcd(M,\ell)}M\right\rfloor,
 \label{eq:iota}
\end{equation}
with \(\gcd(M,0)=M\).

\begin{theorem}[Positive-floor compatibility spectrum and curvature]
\label{thm:G2}
Fix \(M\ge3\) and \(\lambda>0\).
\begin{enumerate}[label=(\roman*)]
\item For every \(k\ge1\) and every choice of \(k\) distinct positive
normalized offset magnitudes,
\begin{equation}
 \dim_{\R}\NM^{(k)}=
 \sum_{\ell=0}^{M-1}\!\left[L-\min\{k,L-\iota_\ell\}\right]
 +\ind_{\{2\mid M\}}\frac M2-1.\label{eq:spectrum}
\end{equation}
\item Every set of \(L\) distinct positive normalized offset magnitudes
certifies persistence, and \(L\) is sharp for every such choice: for
\(0\le k<L\), with
\(\NM^{(0)}:=\EM\), \(\NM^{(k)}\supsetneq\NM(\lambda)\).  The terminal space
is characterized by
\begin{equation}
 g_{q,q+d}=g_{q,q-d}\quad(q\in\mathbb Z_M, 1\le d<M/2),\label{eq:terminal}
\end{equation}
and
\begin{equation}
 \dim\NM(\lambda)=\sum_{d=1}^{L}\gcd(M,d)
 +\ind_{\{2\mid M\}}\frac M2-1.\label{eq:gcd}
\end{equation}
It is independent of the positive magnitude of \(\lambda\).
\item If \(g\in\NM(\lambda)\), then, exactly for every \(q,r>0,t\in\R\),
\begin{align}
 \mu_{q,\lambda}^2(r,tg)&=\mu_{q,\lambda}^2(r,0)
 +4t^2r^2V_{q,r,\lambda}(g),\label{eq:starid}\\
 V_{q,r,\lambda}(g)&=\min_{\alpha\in\R}\sum_{j\ne q}
 \frac{a_{qj}^2(g_{qj}+\alpha)^2}
 {a_{qj}r^2+\lambda/M}.\label{eq:star}
\end{align}
Here \(V\ge0\), with equality iff the star at \(q\) is constant.  Its
restriction to the persistent space has per-node rank
\(\lfloor M/2\rfloor-1\), and \(\sum_qV_{q,r,\lambda}\) is positive definite
on \(\EM\).
\end{enumerate}
\end{theorem}

\begin{proof}
For \(1\le d\le L\), write
\(y_d(q)=g_{\{q,q+d\}}\),
\(a_d=\{4\sin^2(\pi d/M)\}^{-1}\), and
\(\rho_d=\lambda/(Ma_d)\).  The \(a_d\), hence \(\rho_d\), are pairwise
distinct.  After deleting a common nonzero factor, neutrality at radius
\(r_s\) is
\begin{equation}
 \sum_{d=1}^L\frac{\beta_d}{r_s^2+\rho_d}
 \{y_d(q)-y_d(q-d)\}=0,\label{eq:distancebalance}
\end{equation}
where \(\beta_d>0\).  An even diameter has \(b=0\) and contributes exactly
\(M/2\) free unordered edges.

Complexify the real map and apply the DFT in \(q\).  In character \(\ell\),
the \(d\)th column is inactive exactly when \(M\mid\ell d\), giving
\(\iota_\ell\) inactive columns.  On active columns the \(k\)-row matrix is,
up to nonzero column scalings, \((r_s^2+\rho_d)^{-1}\).  Every square minor is
nonzero by the Cauchy determinant \cite{Schechter1959}; hence its rank is
\(\min\{k,L-\iota_\ell\}\).  The real and complex ranks of a real matrix
coincide, so summing the complex block nullities once over all \(M\)
characters, adding the diameter matching, and removing one global constant
proves \eqref{eq:spectrum}.

At \(k=L\), every active block has full column rank, leaving precisely
\(y_d(q)=y_d(q-d)\), which is \eqref{eq:terminal}.  Conversely, that recurrence
kills \eqref{eq:distancebalance} for every radius.  It gives one value on each
of the \(\gcd(M,d)\) cosets generated by \(d\), proving \eqref{eq:gcd}.  If
\(k<L\), character \(\ell=1\) has no inactive column and a kernel of dimension
\(L-k\); adjoining its conjugate character yields a real, nonpersistent
field. Thus sharpness holds for every set of evaluated magnitudes, not merely
generically.

Finally, the pre-profile expression \eqref{eq:mu} is quadratic in
\((t,\kappa)\); its unique minimizer is affine in \(t\).  Persistence removes
the linear coefficient, and substitution \(\kappa=t\alpha\) gives
\eqref{eq:starid}--\eqref{eq:star}.  Positive weights give the equality
condition.  The attainable incident values are indexed by the
\(\lfloor M/2\rfloor\) undirected separations; quotienting a constant star
gives the stated rank.  If every star is constant, shared edges force all
star constants to agree, so the field is quotient-null.
\end{proof}

The first layers illustrate that fixed-radius neutrality can greatly
overestimate genuine persistence:
\begin{equation}
\begin{array}{c|c}
M&\dim\NM^{(1)},\ldots,\dim\NM^{(L)}\\ \hline
5&5,1\\
6&9,5\\
7&14,8,2\\
8&20,13,7.
\end{array}\label{eq:spectrumexamples}
\end{equation}
The displayed dimensions follow directly from \eqref{eq:spectrum}. Exact
rank and numerical conditioning are distinct. For a fixed character with
\(m=L-\iota_\ell\) active distance columns, fixed distinct positive
magnitudes, and \(k\ge m\), a Cauchy--Vandermonde expansion gives
\[
 C_\ell(\lambda)=A_\ell\diag(1,\lambda,\ldots,\lambda^{m-1})B_\ell
 +O(\lambda^m),
 \qquad \lambda\downarrow0,
\]
where \(A_\ell\) has full column rank and \(B_\ell\) is invertible. Hence,
with singular values in decreasing order,
\(\sigma_j(C_\ell(\lambda))\asymp\lambda^{j-1}\), and
\(\kappa_2(C_\ell(\lambda))\asymp\lambda^{-(m-1)}\)
\cite{Demmel2000}. Indeed, fixed invertible row and column transformations
reduce the leading term to a rectangular diagonal matrix, and the
\(O(\lambda^m)\) remainder is of lower order than its smallest nonzero
singular value. The worst block, \(\ell=1\), has \(m=L\). Thus exact
compatibility persists for every \(\lambda>0\), while its conditioning
degenerates at the quantified rate as the positive critical floor approaches
zero. Constants in these comparisons depend on the fixed magnitudes and are
not uniform when they coalesce. At \(\lambda=0\), the poles merge exactly,
which explains why one radius evaluation certifies the larger space in
\eqref{eq:zeroDim}.

At the zero critical replacement floor \(\lambda=0\), all radius rows in
\eqref{eq:D} are proportional. One
positive radius therefore already certifies persistence and, by
\eqref{eq:Drank},
\begin{equation}
 \dim\NM(0)=\binom M2-1-(M-1)=\frac{M(M-3)}2.\label{eq:zeroDim}
\end{equation}
On \(\NM(0)\), the same curvature identity holds, with per-node rank
\(M-3\). The positive-floor Cauchy formula is not evaluated at
\(\lambda=0\); the merged-pole operator has the separate dimension
\eqref{eq:zeroDim}.

For an arbitrary direction the profiling calculation also has the exact
form
\begin{equation}
 \mu_{q,\lambda}^2(r,tg)=\mu_{q,\lambda}^2(r,0)
 +tL_{q,r,\lambda}(g)+t^2Q_{q,r,\lambda}(g),\label{eq:quadraticprofile}
\end{equation}
Thus, along every design direction, the canonical sign strength is quadratic
in the perturbation amplitude; higher-order onsets do not occur.

\section{Sharp Critical Minimax Limit}

Let
\[
 \mathcal R_N^*(\tau,\eps_N)=
 \inf_{U\in\mathcal U_N(\tau)}\inf_{\widehat\theta}
 \sup_{\theta\in\T}\E_{\theta,U}\distT^2(\widehat\theta,\theta),
\]
where \(U\) is deterministic and data-independent, the estimator may use the
known \(U,M,N,\eps_N\), and expectations use the labeled IID categorical law.

\begin{theorem}[Sharp critical design-minimax law]\label{thm:G1}
Fix \(M\ge3\), finite \(\tau\ge0\), and a known sequence
\(0\le\eps_N<1\) satisfying \(\sqrt N\eps_N\to\lambda<\infty\). Assume a
single deterministic unitary shared by all frames, with every unordered pair
coordinate admissible,
uniform independent replacement, and squared circular loss. Then
\begin{equation}
 \boxed{\lim_{N\to\infty}\sqrt N\,\mathcal R_N^*(\tau,\eps_N)
 =\cM(\tau,\lambda).}\label{eq:G1}
\end{equation}
If \([c^*]\) attains \eqref{eq:C}, a least-routing-resource representative
and its phase-aligned tangent produce an asymptotically optimal sequence
\(U_N^*=\exp\{(1-o(1))N^{-1/2}K(c^*)\}U_0\).
\end{theorem}

\emph{Converse.} Any nearly optimal design sequence has, after alignment and
subsequence extraction, an effective \([c]\) in the resource ball. Fix
\((q,r)\), choose the minimizing centre in \eqref{eq:mu}, and consider
\begin{equation}
 \theta_{N,\pm}=\theta_q\pm r\eta_N+\kappa^*\eta_N^2.\label{eq:twopoints}
\end{equation}
Map the signed displacement of any circular estimator from
\(\theta_q+\kappa^*\eta_N^2\), divided by \(\eta_N\), to \([-r,r]\) by metric
projection. For all large \(N\), both parameters lie in one injectivity arc,
so projection cannot increase loss to either endpoint. The action space is
now compact and the rescaled loss bounded. The finite multinomial
log-likelihood ratio for \eqref{eq:twopoints} converges under both signs to
the nuisance-projected scalar Gaussian likelihood; bounded two-experiment
risk transfer \cite{LeCamYang2000,VanderVaart1998} therefore gives
\begin{equation}
 \liminf_N\max_{S=\pm1}\sqrt N\,
 \E_S\distT^2(\widehat\theta_N,\theta_{N,S})
 \ge r^2\psi\{\mu_{q,\lambda}^2(r,[c])\}.\label{eq:twopointlower}
\end{equation}
This statement uses only the sign component proved above, not an unrestricted
Gaussian equivalence. Maximizing \((q,r)\), minimizing over the compact
resource ball, and applying the argument to every convergent subsequence give
\begin{equation}
 \liminf_N\sqrt N\,\mathcal R_N^*(\tau,\eps_N)
 \ge\cM(\tau,\lambda).\label{eq:mainliminf}
\end{equation}

\emph{Achievability.} A complete construction uses
\begin{align}
 m_N&=\lceil N^{3/4}\rceil,&
 R_N&=\{\log(N+e)\}^{1/3},\nonumber\\
 \zeta_N&=N^{-1/32},&
 d_N&=N^{-7/64}.\label{eq:seq}
\end{align}
The guard \(\zeta_N\) is applied uniformly over all limiting replacement
regimes, although it is analytically needed only at the zero critical floor.
Let \(c_N\) range over
a fixed bounded body of routing-minimizing representatives and put
\begin{equation}
 A^{\rm fin}_{N,qj}(r)=(1-\eps_N)a_{qj}r^2+\lambda_N/M.\label{eq:Afinite}
\end{equation}
The exact finite profiled strength is
\begin{equation}
\begin{aligned}
 s_{N,q}(r,c_N)&=\min_{\kappa\in\R}\sum_{j\ne q}
 \frac{(1-\eps_N)^2r^2}{A^{\rm fin}_{N,qj}(r)}\\
 &\quad\times
 \{b_{qj}r^2+2a_{qj}(c_{N,qj}+\kappa)\}^2 .
\end{aligned}\label{eq:sfinite}
\end{equation}
At the sole singular endpoint \(r=\lambda_N=0\), define \(s_{N,q}\) by the
radial limit \(r\downarrow0\) with \(\lambda_N=0\); the product
\(r^2\psi(s_{N,q})\) is zero there.  The weights below
are evaluated only on the guarded branch \(r\ge\zeta_N>0\); below the guard
the rule returns the grid node.  Thus no \(0/0\) expression is evaluated.
No limiting denominator is substituted in \eqref{eq:sfinite}. Denote its
minimizer by \(\kappa_N^*(r,c_N)\) and define
\begin{equation}
 w^{\rm fin}_{N,qj}(r,c_N)=
 \frac{(1-\eps_N)r\{b_{qj}r^2+2a_{qj}(c_{N,qj}+\kappa_N^*)\}}
 {A^{\rm fin}_{N,qj}(r)} .\label{eq:wfinite}
\end{equation}
The exact normal equation is
\begin{equation}
 \sum_{j\ne q}w^{\rm fin}_{N,qj}(r,c_N)a_{qj}=0.\label{eq:finiteorth}
\end{equation}
It removes the leading radial mean and the second-order centre nuisance after
plug-in without imposing a convergence rate on \(\lambda_N\).

Use the pilot sample to obtain \(\widehat q\) and \(\widehat r^2\), with
deterministic smallest-index tie breaking, and define
\begin{equation}
 \mathcal E_N(q,r)=\{\widehat q=q,
 |\widehat r^2-r^2|\le d_N\}.\label{eq:pilotgood}
\end{equation}
Uniformly over the bounded design body and \(0\le r\le3R_N\), Bernstein
bounds give \(\sqrt N\Pr\{\mathcal E_N(q,r)^c\}\to0\). For the independent
main counts \(C_j\), conditionally on the pilot, set
\begin{equation}
 T_N=N^{-1/4}\sum_{j\ne\widehat q}
 w^{\rm fin}_{N,\widehat qj}(\widehat r,c_N)
 \{C_j-(N-m_N)\eps_N/M\}.\label{eq:Tfinite}
\end{equation}

\begin{lemma}[Finite-profile Gaussian approximation and uniform risk
transfer]\label{lem:U4}
Assume only \(\lambda_N=\sqrt N\eps_N\to\lambda\in[0,\infty)\). On the good
pilot event, uniformly in \(q\), the bounded design body, and the guarded
window \(\zeta_N\le r\le3R_N\),
\begin{equation}
 \mathcal L(T_N\mid\mathcal E_N)
 =\mathcal N\{S s_{N,q}(r,c_N),s_{N,q}(r,c_N)\}+o_{\rm BL}(1).
 \label{eq:finiteCLT}
\end{equation}
If \(R_{N,q,r}^{\rm fin}(c_N)\) denotes the normalized risk of the resulting
plug-in rule \(\widehat{rS}=\widehat r\tanh T_N\), with the grid returned
below the universal guard, then
\begin{equation}
 \sup_{\substack{q,\ c_N\in\mathcal K_B\\0\le r\le3R_N}}
 \left|R_{N,q,r}^{\rm fin}(c_N)-
 r^2\psi\{\mu_{q,\lambda}^2(r,[c_N])\}\right|\longrightarrow0.
 \label{eq:risktransfer}
\end{equation}
Here \(\mathcal K_B\) is any fixed compact body of routing-minimizing
representatives. In particular, \eqref{eq:risktransfer} requires no rate for
\(\lambda_N\to\lambda\).
\end{lemma}

\emph{Proof sketch.} A centered one-frame summand in \eqref{eq:Tfinite} is
bounded by \(CN^{-1/4}(R_N+\zeta_N^{-1})\). Its accumulated third absolute
moment is at most
\begin{equation}
 CN^{-1/4}(R_N^5+\zeta_N^{-1})=o(1),\label{eq:BEbound}
\end{equation}
uniformly in \(q,c_N,r\); for a positive limiting floor the singular term is
absent. The exact variance is \(s_{N,q}+o(1)\), the signed mean is
\(Ss_{N,q}+o(1)\), \eqref{eq:finiteorth} cancels the plug-in radial and centre
terms, and \(d_N\zeta_N^{-3}\to0\). Berry--Esseen smoothing, including the
degenerate-variance endpoint, proves \eqref{eq:finiteCLT}; the chosen rates
make its bounded-Lipschitz error \(o(R_N^{-2})\). Supplementary Lemma~8 gives
an explicit small-variance threshold and observes that the squared-loss test
function has bounded-Lipschitz norm \(O(R_N^2)\), so this rate is sufficient
for the risk.

For the transfer, fix \(0<\delta<R<\infty\) and split
\([0,3R_N]=[0,\delta]\cup[\delta,R]\cup[R,3R_N]\). First,
\begin{equation}
 \sup_{r\le\delta}|r^2\psi(s_{N,q})-
 r^2\psi(\mu_{q,\lambda}^2)|\le2\delta^2.\label{eq:smallzone}
\end{equation}
Second, the denominators are bounded away from their singular endpoint on
\([\delta,R]\); compact continuity gives
\begin{equation}
 \sup_{\substack{q,c_N\in\mathcal K_B\\\delta\le r\le R}}
 |s_{N,q}(r,c_N)-\mu_{q,\lambda}^2(r,[c_N])|\to0,\label{eq:midzone}
\end{equation}
and uniform continuity of \(\psi\) transfers the compact risk. Third,
antisymmetric pairing of \(b_{qj}\) makes the two profile centres uniformly
bounded, and cubic dominance gives, for all sufficiently large fixed \(R\),
\begin{equation}
 s_{N,q}(r,c_N)\wedge\mu_{q,\lambda}^2(r,[c_N])\ge cr^4,
 \quad R\le r\le3R_N.\label{eq:largecoercive}
\end{equation}
The hard-sign rule yields \(\psi(s)\le Ce^{-cs}\), hence both large-zone
risks are bounded by \(Cr^2e^{-cr^4}\). Take first \(N\to\infty\), then
\(\delta\downarrow0\), and finally \(R\uparrow\infty\). This proves
\eqref{eq:risktransfer}. Uniform strength convergence fails at \(r=0\), so
risk transfer is established directly; the factor \(r^2\) suppresses the
nonuniform small-radius regime.

The choices in \eqref{eq:seq} satisfy
\[
 N^{-1/4}R_N^5\to0,\quad d_N\zeta_N^{-3}\to0,\quad
 \sqrt N(N+1)^Me^{-cR_N^4}\to0.
\]
They simultaneously control the growing-window Taylor remainder, pilot
error, singular zero-critical-floor weights, conditional triangular approximation,
and every splice error.

\begin{lemma}[Uniform outer Hellinger inverse]\label{lem:U5}
Let \(\Theta_N^{\rm out}\) be the circle after removing the open
\(R_N\eta_N\)-neighborhood of every routing node. Uniformly over bounded
near-router tangents, write \(P_{q,h,N}\) for the law at
\(\theta_q+h\). Then
\begin{align}
 H^2(P_{q,h_1,N},P_{q,h_2,N})&\ge c|h_1-h_2|^2
 &&\text{on one side},\label{eq:Hsameside}\\
 H^2(P_{q,h,N},P_{q,-h,N})&\ge ch^4
 &&R_N\eta_N\le|h|\le h_0.\label{eq:Hopposite}
\end{align}
On each fixed compact set away from the nodes there are \(c,\rho>0\) such
that
\begin{equation}
 H^2(P_{\theta,N},P_{\theta',N})\ge
 c\distT(\theta,\theta')^2
 \quad\text{if }\distT(\theta,\theta')\le\rho,\label{eq:Hcompactmetric}
\end{equation}
and pairs farther apart have a uniform positive Hellinger separation.
Consequently the minimum-Hellinger estimator on \(\Theta_N^{\rm out}\) has
uniform risk \(O(N^{-1})\).
\end{lemma}

\emph{Proof sketch.} For \eqref{eq:Hopposite}, use the fixed non-diameter
cell \(j(q)=q+1\): its odd probability part is
\(2b_{qj}h^3+4a_{qj}c_{qj}N^{-1/2}h+O(h^5+N^{-1/2}h^3+N^{-1}h)\), and
\(b_{q,q+1}\ne0\). The selected square-root coordinate has a derivative
bounded away from zero on either side, which gives \eqref{eq:Hsameside}.
On compact sets separated from the routing nodes, the exact-router first moment
satisfies
\begin{equation}
 m_1'(\theta)=\frac{i(M-1)}M e^{i\theta}(1-e^{-iM\theta}).\label{eq:m1prime}
\end{equation}
It never vanishes away from the routing nodes. All exact-router probabilities
are then positive, so the derivative of the square-root probability embedding
is nonzero. Compactness gives a uniform positive minimum, Taylor's theorem
gives \eqref{eq:Hcompactmetric}, and exact-router injectivity gives the
nonlocal separation. The \(C^1\) convergence
\((U_N,\eps_N)\to(U_0,0)\) transfers both constants to the near-router
family. Finally, for the empirical law \(\widehat p\), the method of types
\cite{CoverThomas2006} gives
\begin{equation}
 \Pr_p\{H^2(\widehat p,p)\ge u\}\le(N+1)^Me^{-Nu}\label{eq:typebound}
\end{equation}
This bound
integrates with \eqref{eq:Hsameside}--\eqref{eq:Hcompactmetric} to
\(O(N^{-1})\) metric risk; \eqref{eq:Hopposite} controls side confusion by
\((N+1)^Me^{-cNh^4}\).

The local rule uses \(\widehat q,\widehat r,T_N\); the outer rule is
\begin{equation}
 \widehat\theta_{\rm out}\in\argmin_{\phi\in\Theta_N^{\rm out}}
 H(\widehat p,\widetilde p_{\phi,N}^{U_N}),\label{eq:outerest}
\end{equation}
with smallest-angle tie breaking. Invoke the local rule when the pilot finds
a routed node and \(\widehat r\le2R_N\), and invoke \eqref{eq:outerest}
otherwise. Compact argmin measurability and deterministic tie rules make this
one measurable estimator depending only on \(N,M,U_N,\eps_N\), never on the
true node, radius, or limiting regime. The complement of \(\mathcal E_N\)
is negligible at the \(N^{1/2}\)-risk scale. Lemma~\ref{lem:U4} controls the
local band, Lemma~\ref{lem:U5} the outer band, and either rule has
\(o(N^{-1/2})\) risk on the transition band. Hence
\begin{equation}
 \limsup_N\sqrt N\,\mathcal R_N^*(\tau,\eps_N)
 \le\cM(\tau,\lambda).\label{eq:mainlimsup}
\end{equation}
The converse and achievability bounds complete the proof of
Theorem~\ref{thm:G1} without requiring a convergence rate for
\(\lambda_N-\lambda\).

\section{Consequences for Design Rigidity and Floor Dependence}

\subsection{The global three/four boundary}

\begin{theorem}[Objective-specific global rigidity]\label{thm:M3global}
For every finite \(\lambda\ge0\) and \(\tau\ge0\),
\begin{equation}
 \cM(\tau,\lambda)=\cM(0,\lambda),\qquad M=3,\label{eq:M3rigid}
\end{equation}
and \([c]=0\) is the unique global minimizer in the canonical quotient.  In
contrast, for every fixed \(M\ge4\), finite \(\lambda\ge0\), and \(\tau>0\),
\begin{equation}
 \cM(\tau,\lambda)<\cM(0,\lambda).\label{eq:M4improve}
\end{equation}
``Global'' in \eqref{eq:M3rigid} refers to the complete critical quotient of
Theorem~\ref{thm:G1}, not to arbitrary finite-distance unitaries.
\end{theorem}

\begin{proof}
For \(M=3\), put
\(\delta_q=c_{q,q+1}-c_{q,q-1}\).  Direct profiling of the two equally
weighted dark cells gives, for \(x=r^2>0\),
\begin{equation}
 \begin{aligned}
 s_\lambda(x,\delta)&=\frac{2x}{3(x+\lambda)}
 \left(\delta+\frac{x}{\sqrt3}\right)^2,\\
 G_\lambda(\delta)&=\sup_{x\ge0}x\psi\{s_\lambda(x,\delta)\}.
 \end{aligned}\label{eq:M3profile}
\end{equation}
where the risk at \(x=0\) is zero.  Hence
\(F_{3,\lambda}([c])=\max_qG_\lambda(\delta_q)\).

The required strict envelope ordering follows as follows. The supremum in
\eqref{eq:M3profile} is positive and attained in \((0,\infty)\); its
maximizers remain in one compact interval when \(\delta\) ranges locally.
Every maximizer satisfies \(\delta+x/\sqrt3>0\).  Indeed, if \(\delta<0\)
and \(x_0=-\sqrt3\delta\), then the objective is below \(x_0\) for
\(x<x_0\), equals \(x_0\) at \(x_0\), and equals
\(x_0+\epsilon+O(\epsilon^2)>x_0\) immediately to its right.  At every
maximizer,
\begin{equation}
 \partial_\delta\{x\psi(s_\lambda)\}
 =\frac{4x^2}{3(x+\lambda)}
 \left(\delta+\frac{x}{\sqrt3}\right)\psi'(s_\lambda)<0.\label{eq:M3deriv}
\end{equation}
A maximizing-sequence argument on the local compact set therefore makes the
upper right Dini derivative of \(G_\lambda\) strictly negative.  The Dini
mean-value lemma proves that \(G_\lambda\) is strictly decreasing.  Finally,
\(\sum_q\delta_q=0\), and all \(\delta_q=0\) exactly for a constant edge
field.  Every nonzero quotient has some \(\delta_q<0\), whence
\(F_{3,\lambda}([c])\ge G_\lambda(\delta_q)>G_\lambda(0)\).  This proves
\eqref{eq:M3rigid} and quotient uniqueness.

For \(M\ge4\), let \(g\) equal one on the nearest-neighbor cycle and zero on
all other edges.  The \(+1\) and \(-1\) terms in \eqref{eq:D} cancel at every
node and radius, for both zero and positive floors, so \(g\) is persistent.
Each star contains two ones and at least one zero.  For positive floors,
\(V_{q,r,\lambda}(g)>0\) and \eqref{eq:starid} strictly decreases every
local risk for nonzero \(t\).  At zero floor the identical curvature formula
follows directly from \eqref{eq:qdist} for \(r>0\), where all star weights
remain positive.  The redesigned maximum is
attained at a positive radius, so the pointwise inequality is strict also
after maximization.  Since
\(\|t[g]\|_{\rm rt}^2=t^2\|[g]\|_{\rm rt}^2\), every \(\tau>0\) admits such
a feasible \(t\), proving \eqref{eq:M4improve}.

\end{proof}

\subsection{Arithmetic localization of curvature visibility}

For a nonzero persistent class, call node \(q\) \emph{curvature-visible} when
its profiled star curvature is positive, and define
\begin{equation}
 S([g])=\{q:V_{q,r,\lambda}(g)>0\}.\label{eq:visibility}
\end{equation}
The constant-star equality condition shows that this set is gauge invariant
and, within either floor regime, independent of \(r>0\).

\begin{proposition}[Smallest-prime curvature-visibility law]
\label{prop:localization}
Let \(M\ge4\), \(\lambda>0\), and let \(p_{\min}(M)\) be the smallest prime
dividing \(M\).  Then
\begin{equation}
 \min_{0\ne[g]\in\NM(\lambda)}|S([g])|=p_{\min}(M).\label{eq:pmin}
\end{equation}
For odd \(M\ge5\), every nonzero positive-floor persistent mode is visible at
all nodes iff \(M\) is prime.  At zero floor the corresponding minimum is two
for even \(M\ge4\), and three for odd \(M\ge5\); \(M=3\) is excluded because
\(\NM(0)=\{0\}\).
\end{proposition}

\begin{proof}
If \(S([g])\) is all nodes, the lower bound in \eqref{eq:pmin} is immediate.
Otherwise choose an invisible node and gauge its constant star to zero.  Any
other invisible star also has value zero because it shares an edge with the
first.  Thus every nonzero edge lies wholly inside \(S([g])\).  By
\eqref{eq:terminal}, a nonzero distance-\(d\) edge propagates around an orbit
of \(M/\gcd(M,d)\) nodes, at least \(p_{\min}(M)\); an even diameter has two
endpoints and corresponds to \(p_{\min}=2\).  This proves the lower bound.
One diameter edge attains it for even \(M\).  For odd \(M\ge5\), put
\(d=M/p_{\min}(M)\) and set one full \(d\)-orbit to one and all other edges
to zero.  Its visible support is exactly that orbit, proving attainment and
the prime corollary.

At zero floor, orient active edges and set \(h_{qj}=b_{qj}g_{qj}\).  Then
persistence is ordinary divergence freedom.  After the same invisible-star
gauge, a nonzero active support has no degree-one vertex and therefore
contains at least three vertices; a triangle circulation attains three for
odd \(M\ge5\).  For even \(M\), one zero-weight diameter edge attains two.
\end{proof}

\subsection{Geometry/value decoupling at the zero critical replacement floor}

Choose any Euclidean realization of \(\EM\), and let \(P_0\) and \(P_+\) be
the orthogonal projectors onto \(\NM(0)\) and the common positive-floor space,
respectively.

\begin{proposition}[Discontinuous kernels, continuous value]\label{prop:decouple}
The spaces satisfy \(\NM(\lambda)\subseteq\NM(0)\) for every \(\lambda>0\),
and
\begin{equation}
 \|P_0-P_+\|_{\rm op}=
 \begin{cases}0,&M=3,4,\\1,&M\ge5.\end{cases}\label{eq:projectorjump}
\end{equation}
Nevertheless \((\tau,\lambda)\mapsto\cM(\tau,\lambda)\) is jointly
continuous at every finite point of \([0,\infty)^2\), nonincreasing in
\(\tau\), and nondecreasing in \(\lambda\).  No strict monotonicity in
\(\lambda\) is asserted.
\end{proposition}

\begin{proof}
Condition \eqref{eq:terminal} kills the zero-critical-floor divergence, giving the
inclusion.  Equations~\eqref{eq:gcd} and \eqref{eq:zeroDim} show equality for
\(M=3,4\) and proper inclusion for every \(M\ge5\).  For proper nested
finite-dimensional subspaces, \(P_0-P_+\) is the identity on the nonzero
space \(\NM(0)\cap\NM(\lambda)^\perp\), proving \eqref{eq:projectorjump}.

For value continuity, restrict designs to any bounded routing ball and
\(\lambda\) to a compact interval.  Uniformly there, split radii into
\([0,\delta]\), \([\delta,R]\), and \([R,\infty)\).  The first-zone risk is
at most \(\delta^2\).  On the middle zone, all denominators in
\eqref{eq:mu} stay away from their sole singular endpoint and compact
continuity applies.  In the last zone, paired antisymmetry of \(b\) bounds the
profile centre and cubic dominance gives
\(\mu^2\ge cr^4\), hence
\(R_{q,r,\lambda}\le Cr^2e^{-cr^4}\).  Taking first parameter limits, then
\(\delta\downarrow0\), and finally \(R\uparrow\infty\) proves uniform
continuity of \(F_{M,\lambda}\) on bounded design balls, including at
\(\lambda=0\).  Reparameterizing the compact feasible ball by
\([c]=\sqrt\tau[d]\), \(\|[d]\|_{\rm rt}\le1\), proves joint continuity of
its minimum.  Feasible-set inclusion gives monotonicity in \(\tau\); increasing
\(\lambda\) decreases every \(\mu^2\) in \eqref{eq:mu}, and strict decrease of
\(\psi\) gives the stated weak information ordering after minimization,
consistent with Blackwell degradation \cite{Blackwell1953}.
\end{proof}

The endpoint is governed by the boundary-layer ratio \(\lambda/r^2\). For
fixed \([c]\) and \(\lambda=\rho r^2\), direct profiling in \eqref{eq:mu}
gives
\[
 \mu_{q,\rho r^2}^2(r,[c])\longrightarrow
 4\min_{\kappa\in\R}\sum_{j\ne q}
 \frac{a_{qj}^2(c_{qj}+\kappa)^2}{a_{qj}+\rho/M},
 \qquad r\downarrow0.
\]
The limit generally depends on \(\rho\), whereas
\(0\le R_{q,r,\lambda}\le r^2\). This path dependence explains how the raw
profiled strength and the compatibility kernel can be singular while the
risk and optimized value remain continuous.

\paragraph{Closed-form reference value.}
Define
\(K_M=(M^2-1)(M^2-4)/180\) and
\(\Gamma_\psi=\sup_{s\ge0}\sqrt{s}\,\psi(s)\).  The trigonometric identity
\(\sum_{j\ne q}b_{qj}^2/a_{qj}=K_M\) gives the exact no-resource,
zero-critical-floor value
\begin{equation}
 \cM(0,0)=\sqrt{\frac{180}{(M^2-1)(M^2-4)}}\,\Gamma_\psi.\label{eq:benchmark}
\end{equation}
Theorem~\ref{thm:M3global} makes the same value valid for every finite
three-channel resource and provides a closed-form reference for numerical
evaluation of the canonical minimax functional.

\section{Scope and Limitations}

The theorem does not cover growing \(M\), unknown or nonuniform replacement,
adaptive or random transforms, sparse coupling graphs, nonunitary channelizers,
multi-tone signals, calibration error, dependent frames, or full-vector or
power observations. The categorical interface should not be mistaken for a
conventional receiver without an explicit randomized reporting mechanism.
The attaining estimator is a statistical existence and attainment construction;
its computational complexity is not analyzed, and no claim of algorithmic
efficiency is made.

\section{Conclusion}

Near-perfect DFT routing imposes an exact statistical price in the labeled
index-only observation model. The sharp minimax theorem identifies that
price, proves attainment over the complete critical shrinking-defect class,
and supplies one measurable estimator valid through the singular zero
critical replacement floor. The compatibility theorem resolves the
first-order-neutral design geometry through Fourier-character Cauchy blocks:
\(\lfloor(M-1)/2\rfloor\) distinct normalized offset evaluations are both
necessary and sufficient for universal persistence certification, and the
terminal quotient has an explicit arithmetic dimension and positive-definite
aggregate curvature. These results establish a sharp channel-count boundary:
the zero tangent uniquely minimizes the canonical critical functional for
\(M=3\), whereas every positive routing-defect budget admits a strictly
improving design for \(M\ge4\). They further identify a smallest-prime
curvature-visibility scale and, for \(M\ge5\), a maximal compatibility-kernel
jump at \(\lambda=0\) without a jump in minimax value. The conclusions rely
essentially on fixed \(M\), known uniform replacement, labeled categorical
reports, and one deterministic shared unitary; alternative observation or
adaptation models require separate analysis.

\section*{Supplementary Material}
\setcounter{equation}{55}
\setcounter{theorem}{7}
\setcounter{section}{0}
\renewcommand{\thesection}{S\arabic{section}}
\renewcommand{\thesubsection}{\thesection-\Alph{subsection}}
\section{Complete finite experiment and notation}

Fix an integer \(M\ge3\). Let \(\T=\R/(2\pi\mathbb Z)\),
\(\theta_q=2\pi q/M\), and
\[
 x_\theta[n]=M^{-1/2}e^{in\theta},\qquad 0\le n<M.
\]
For a deterministic \(U\in U(M)\), one observation has probabilities
\(p_\theta^U(j)=|(Ux_\theta)_j|^2\). Known independent uniform replacement
at level \(\eps_N\) gives
\[
 \widetilde p_{\theta,N}^U=(1-\eps_N)p_\theta^U+\eps_N\one/M.
\]
There are \(N\) labeled IID observations, and loss is squared circular
distance. We assume throughout
\[
 \lambda_N=\sqrt N\eps_N\longrightarrow\lambda\in[0,\infty),
 \qquad \delta_{\rm ch}(U_N)\le\tau/N,
\]
for fixed finite \(\tau\). No rate for \(\lambda_N-\lambda\) is assumed.

The reference DFT channelizer is \(U_0=X^\dagger\), where
\(X=[x_{\theta_0},\ldots,x_{\theta_{M-1}}]\). Write
\(\alpha_{qj}=\pi(q-j)/M\) and, for \(j\ne q\),
\begin{equation}
 a_{qj}=\frac1{4\sin^2\alpha_{qj}},\qquad
 b_{qj}=-\frac{\cos\alpha_{qj}}{4\sin^3\alpha_{qj}}.
\label{supp:eq:ab}
\end{equation}
Then \(a_{qj}=a_{jq}>0\), \(b_{qj}=-b_{jq}\), and an even-\(M\) diameter
edge has \(b=0\). Every orientation argument below uses
\(j=q+1\), whose cubic coefficient is nonzero for all \(M\ge3\).

\subsection{Boundary facts used below}

The finite geometric sum gives the exact Dirichlet probabilities and, by
Taylor expansion, \(p_{qj}(h)=a_{qj}h^2+b_{qj}h^3+O(h^4)\) for \(j\ne q\)
and \(p_{qq}(h)=1-(M^2-1)h^2/12+O(h^4)\).  It also gives the exact first
label moment
\[
 \sum_j e^{2\pi ij/M}p_\theta^{U_0}(j)
 =\{(M-1)e^{i\theta}+e^{-i(M-1)\theta}\}/M.
\]
Equality in the triangle inequality proves global injectivity for \(M\ge3\);
\(M=2\) instead has the exact collision
\(p_{q,q+1}(h)=p_{q,q+1}(-h)\).  For \(M\ge3\), the cell \(j=q+1\) has
nonzero \(b\), so its opposite-sign probability difference is
\(2b_{q,q+1}h^3+O(h^5)\) against mass \(2a_{q,q+1}h^2+O(h^4)\).
Consequently the squared Hellinger separation is \(\asymp h^4\), fixing the
\(N^{-1/4}\) local and \(N^{-1/2}\) replacement scales used below.

\subsection{Near-router alignment, edge reduction, and resource}

Define
\[
 \delta_{\rm ch}(U)=\min_{\pi\in S_M}\max_q
 \{1-|(Ux_{\theta_q})_{\pi(q)}|^2\}.
\]

If \(\delta_{\rm ch}(U_N)\le\tau/N\), phase/permutation alignment makes every
diagonal squared modulus at least \(1-\tau/N\).  Unitarity then gives
\(\|A_N-I\|_F=O(N^{-1/2})\) for \(A_N=U_NU_0^\dagger\), so the principal
unitary logarithm yields
\[
 A_N=e^{N^{-1/2}K_N},\qquad K_N^\dagger=-K_N,\qquad
 \sup_N\|K_N\|<\infty.
\]
Differentiation gives
\(v_{qj}=e^{-i\alpha_{qj}}/(2\sin\alpha_{qj})\),
\(|v_{qj}|^2=a_{qj}\), and \(v_{jq}=-\overline v_{qj}\).  Hence
\[
 c_{qj}=\Re(v_{qj}^*K_{jq})/a_{qj}=c_{jq}\in\R.
\]
Only this coordinate enters the critical cross term.  Cauchy--Schwarz gives
\(|K_{jq}|^2\ge a_{qj}c_{qj}^2\), with equality for
\(K(c)_{jq}=c_{qj}v_{qj}\) for \(j\ne q\), with \(K(c)_{qq}=0\); the reverse
entry is its negative conjugate.
Thus unused quadratures consume resource without changing the limit.

Use the localization
\[
 \theta=\theta_q+Sr\etaN+\kappa\etaN^2,\qquad \etaN=N^{-1/4}.
\]
The signed coefficient is
\(Sr\{b_{qj}r^2+2a_{qj}(c_{qj}+\kappa)\}\). Hence a common edge translation
is absorbed by \(\kappa\), and the effective space is
\(\EM=\R^{E(K_M)}/\operatorname{span}\{\one\}\). Define
\begin{equation}
 \|[c]\|_{\rm rt}^2=
 \min_{\beta\in\R}\max_q\sum_{j\ne q}a_{qj}(c_{qj}+\beta)^2.
\label{supp:eq:rtnorm}
\end{equation}
Writing
\(f_c(\beta)=\max_q\sum_{j\ne q}a_{qj}(c_{qj}+\beta)^2\) and
\(A_\Sigma=(M^2-1)/12\), the function
\(f_c(\beta)-A_\Sigma\beta^2\) is convex. Its unique minimizer
\(\beta_\ast\) satisfies
\[
 f_c(\beta)\ge f_c(\beta_\ast)+A_\Sigma|\beta-\beta_\ast|^2.
\]
The nullspace consists exactly of constants. This routing-minimizing gauge
need not be Euclidean orthogonal.

For bounded \(c\), expansion of the routed columns gives
\begin{equation}
 \delta_{\rm ch}(e^{tK(c)}U_0)=
 t^2\max_q\sum_{j\ne q}a_{qj}c_{qj}^2+O(t^3).\label{supp:eq:defect}
\end{equation}
Thus \(t=N^{-1/2}\) is equivalent to an \(N^{-1}\) defect.  The preceding
alignment and finite-dimensional subsequence compactness put every feasible
limit in the quotient resource ball.  Shrinking a boundary tangent by
\(1-o(1)\) enforces the exact finite constraint without changing its limit.

\section{Proof of the Critical Minimax Limit}

For \(U_N=e^{\etaN^2K(c)}U_0\), write
\begin{align*}
 A^{\rm fin}_{N,qj}(r)&=(1-\eps_N)a_{qj}r^2+\lambda_N/M,\\
 A_{qj}(r,\lambda)&=a_{qj}r^2+\lambda/M.
\end{align*}

Uniformly in \(q,j\ne q,S\), and bounded \(r,\kappa,c\), analytic expansion
of the amplitude and then its squared modulus gives
\begin{equation}
 \widetilde p_{\theta,N}^{U_N}(j)
 =\etaN^2A^{\rm fin}_{N,qj}(r)
 +S\etaN^3d_{qj}(r,c,\kappa)+o(\etaN^3),\label{supp:eq:cell}
\end{equation}
where \(d_{qj}=r\{b_{qj}r^2+2a_{qj}(c_{qj}+\kappa)\}\).  Define
\(\gamma_{qj}=\left.\partial_h^2(U_0x_{\theta_q+h})_j\right|_{h=0}\), so that
\(\Re(v_{qj}^*\gamma_{qj})=b_{qj}\).  The amplitude is
\[
 v_{qj}h+\tfrac12\gamma_{qj}h^2+\etaN^2K(c)_{jq}
 +O(h^3+\etaN^2|h|+\etaN^4);
\]
substitute \(h=Sr\etaN+\kappa\etaN^2\), use
\(\Re(v_{qj}^*K(c)_{jq})=a_{qj}c_{qj}\), and retain the exact
\(1-\eps_N\) factor in the even term.

Let \(C_{N,j}\) be the dark counts and define
\(Z_{N,j}=N^{-1/4}\{C_{N,j}-N\etaN^2A^{\rm fin}_{N,qj}(r)\}\).

On compact local sets, the dark-count vector satisfies
\[
 (Z_{N,j})_{j\ne q}\Longrightarrow
 N\left(Sd_q,\diag\{A_{qj}(r,\lambda)\}_{j\ne q}\right),
\]
uniformly in those compact parameters.  To verify this, expand the one-frame
log characteristic function at \(N^{-1/4}u\).  Its accumulated quadratic
term tends to \(-\tfrac12\sum_jA_ju_j^2\), mixed multinomial covariances
vanish after scaling, and the accumulated third term is \(O(N^{-1/4})\).
The mean follows from \eqref{supp:eq:cell}; Cram\'er--Wold finishes.  Only this
sign component, not unrestricted experiment equivalence, is claimed.

Weighted projection of the centre direction yields
\begin{equation}
 \mu_{q,\lambda}^2(r,[c])=\min_{\kappa\in\R}
 \sum_{j\ne q}\frac{r^2\{b_{qj}r^2+2a_{qj}(c_{qj}+\kappa)\}^2}
 {a_{qj}r^2+\lambda/M}.
\label{supp:eq:mu}
\end{equation}
At \((r,\lambda)=(0,0)\), use the radial limit \(r\downarrow0\) with
\(\lambda=0\); no joint extension of \(\mu^2\) is asserted, and the chosen
endpoint value is multiplied by \(r^2=0\).
The projected scalar law is \(Y=\mu S+Z\), \(Z\sim N(0,1)\). For the
equal-prior binary problem, the posterior mean is
\(\widehat{rS}=r\tanh(\mu Y)\), and the common conditional risk is
\begin{equation}
 r^2\psi(\mu^2),\qquad
 \psi(s)=\E[\sech^2(s+\sqrt sZ)].\label{supp:eq:psi}
\end{equation}
The posterior-variance identity cited in the main paper gives
\(\psi'(s)=-\E[\operatorname{Var}(S\mid\sqrt sS+Z)^2]<0\) for \(s>0\).

Define
\begin{align}
 F_{M,\lambda}([c])&=\max_q\sup_{r\ge0}
 r^2\psi(\mu_{q,\lambda}^2(r,[c])),\label{supp:eq:F}\\
 \cM(\tau,\lambda)&=
 \min_{\|[c]\|_{\rm rt}^2\le\tau}F_{M,\lambda}([c]).\label{supp:eq:C}
\end{align}
At small radius, the risk is at most \(r^2\). At large radius,
\(\kappa^*=O(1)\) uniformly on bounded quotient balls and
\[
 \mu^2=r^4\sum_{j\ne q}b_{qj}^2/a_{qj}+O(r^2).
\]
The leading constant is positive for \(M\ge3\), so a hard-sign bound gives
exponential decay of \(\psi\). Consequently the radius supremum is uniformly
compact, \(F\) is continuous, and the finite-dimensional resource-ball
minimum is attained.

\subsection{Compact-action converse}

Take an arbitrary nearly optimal sequence and pass to a subsequence on which
the aligned effective fields \(c_N\) converge to \(c\). By \eqref{supp:eq:defect}
and lower semicontinuity, \(\|[c]\|_{\rm rt}^{2}\le\tau\). Fix \(q,r>0\), let
\(\kappa^*\) minimize \eqref{supp:eq:mu}, put
\(\theta_{N,S}=\theta_q+Sr\etaN+\kappa^*\etaN^2\), and write
\(A_{N,j}=A^{\rm fin}_{N,qj}(r)\), \(d_j=d_{qj}(r,c,\kappa^*)\), and let
\(d_{N,j}\) be its exact finite-\(N\) counterpart. Bounded-tangent
analyticity gives \(d_{N,j}\to d_j\), but no rate is used. There is a common part
\(\bar p_{N,j}=\etaN^2A_{N,j}+O(\etaN^4)\) such that
\(p_{N,j,S}=\bar p_{N,j}+S\etaN^3d_{N,j}+O(\etaN^5)\). Routed-cell cancellation
and a second-order log expansion therefore give, under either sign,
\[
\begin{aligned}
 \ell_N&:=\log\frac{dP_{N,+}}{dP_{N,-}}
 =2\etaN\sum_{j\ne q}\frac{d_{N,j}}{A_{N,j}}\\
 &\quad{}\times\{C_{N,j}-N\etaN^2A_{N,j}\}
   +o_{P_{N,S}}(1),\\
 \ell_N&\Longrightarrow \mathcal N(2S\mu^2,4\mu^2).
\end{aligned}
\]
Here \(A_{N,j}>0\), including when \(\lambda=0\), because \(r>0\); the accumulated
dark-cell cubic remainder is \(O(N\etaN^5)=o(1)\).  Thus the binary likelihood
ratio experiment converges, not merely its count statistic.

Project any rescaled circular action onto \([-r,r]\).  On the common
injectivity arc this cannot increase either endpoint loss, so the action and
loss are bounded.  Bounded binary-experiment risk convergence makes the
limiting average risk at least the Gaussian Bayes value
\(r^2\psi(\mu_{q,\lambda}^2(r,[c]))\), while maximum risk dominates average
risk.  Taking the supremum over \(q,r\), and using
\(F_{M,\lambda}([c])\ge\cM(\tau,\lambda)\), proves along every subsequence
\begin{equation}
 \liminf_N\sqrt N\,\mathcal R_N^*(\tau,\eps_N)
 \ge\cM(\tau,\lambda).\label{supp:eq:lower}
\end{equation}

\subsection{Uniform constructive upper bound}

This section proves the global achievability statement without relying on an
unspecified diagonal extraction. Put
\begin{equation}
 \etaN=N^{-1/4},\quad m_N=\lceil N^{3/4}\rceil,
 \quad R_N=\{\log(N+e)\}^{1/3}.
\label{supp:eq:sequences}
\end{equation}
For every sequence, also put
\begin{equation}
 \zeta_N=N^{-1/32},\qquad d_N=N^{-7/64};\label{supp:eq:guard}
\end{equation}
the guard is analytically essential only at zero limiting floor, but applying
it universally prevents the estimator from using the limiting regime as an
oracle. Direct calculation gives
\begin{align}
 N^{-1/4}R_N^5&\to0,\qquad N^{-1/4}\zeta_N^{-1}\to0,\nonumber\\
 d_N\zeta_N^{-3}&=N^{-1/64}\to0,\label{supp:eq:rates}\\
 d_N&=o(\zeta_N^2),\nonumber\\
 \sqrt N(N+1)^Me^{-cR_N^4}&\to0.\nonumber
\end{align}
The last limit holds for every \(c>0\), because
\(R_N^4=(\log(N+e))^{4/3}\gg\log N\).

\subsubsection{Growing-window expansion}

Let \(U_N=e^{\etaN^2K_N}U_0\), \(\sup_N\|K_N\|<\infty\), and
\(\theta=\theta_q+Sr\etaN+\kappa\etaN^2\), with
\(|r|,|\kappa|\le3R_N\). Analyticity of the finite Fourier sum and matrix
exponential, uniformly on the bounded tangent body, gives
\begin{equation}
 \widetilde p_{N,j}=\etaN^2A^{\rm fin}_{N,qj}(r)
 +S\etaN^3d_{N,qj}(r,c_N,\kappa)+\mathcal E_{N,qj},\label{supp:eq:grow}
\end{equation}
where exact finite signal/replacement factors may be retained in
\(d_{N,qj}\), and
\begin{equation}
 \sup|\mathcal E_{N,qj}|\le CN^{-1}(1+R_N^4).\label{supp:eq:growrem}
\end{equation}
The monomials are \(h^4,\etaN^2h^2,\etaN^4\), and centre remainders. When
inserted into the weighted contrast below, \eqref{supp:eq:growrem} contributes at
most
\[
 CN^{-1/4}(1+R_N^4)(R_N+\zeta_N^{-1})=o(1)
\]
by \eqref{supp:eq:rates}.

\subsubsection{Pilot node and squared radius}

Use the first \(m_N\) reports as an independent pilot. The routed cell has
probability \(1-O(\etaN^2R_N^2)\), while every dark cell has probability
\(O(\etaN^2R_N^2)\). Thus the routed-count maximizer, with the smallest-index
tie rule, satisfies uniformly on \(r\le3R_N\)
\begin{equation}
 \Pr(\widehat q\ne q)\le Ce^{-cm_N}.\label{supp:eq:qpilot}
\end{equation}
Let \(L_N\) be the pilot count outside \(\widehat q\), and
\(A_\Sigma=\sum_{j\ne q}a_{qj}=(M^2-1)/12\). Define
\begin{equation}
 \widehat r^2=\left[
 \frac{L_N/m_N-\eps_N(1-1/M)}{(1-\eps_N)\etaN^2A_\Sigma}
 \right]_+ .\label{supp:eq:rhat}
\end{equation}
Bernstein's inequality and \eqref{supp:eq:grow} give
\begin{equation}
 \Pr(|\widehat r^2-r^2|>d_N)
 \le C\exp\left\{-c\frac{m_N\etaN^2d_N^2}{1+R_N^2}\right\}+Ce^{-cm_N}.
\label{supp:eq:rpilot}
\end{equation}
The exponent is \(N^{1/32}/(1+R_N^2)\to\infty\). The deterministic bias
after division in \eqref{supp:eq:rhat} is
\(O(\etaN R_N^2+\etaN^2R_N^4)=o(d_N)\). A deviation of order \(R_N^2\)
likewise gives false-local and false-outer probabilities at most
\begin{equation}
 C\exp\{-cN^{1/4}R_N^2\}.\label{supp:eq:pilotbranch}
\end{equation}

\subsubsection{Exact finite-sample orthogonality}

The construction retains \(\lambda_N\) because substitution of its limit is
not uniform on the growing window. For a fixed upper-bound field \(c_N\),
profile the exact finite
quadratic having denominators
\[
 A^{\rm fin}_{N,qj}(r)=(1-\eps_N)a_{qj}r^2+\lambda_N/M
\]
and define its exact profiled strength
\begin{equation}
\begin{aligned}
 s_{N,q}(r,c_N)&=\min_{\kappa\in\R}\sum_{j\ne q}
 \frac{(1-\eps_N)^2r^2}{A^{\rm fin}_{N,qj}(r)}\\
 &\quad\times\{b_{qj}r^2+2a_{qj}(c_{N,qj}+\kappa)\}^2 .
\end{aligned}\label{supp:eq:sfinite}
\end{equation}
At \(r=\lambda_N=0\), \(s_{N,q}\) denotes the radial limit
\(r\downarrow0\) with \(\lambda_N=0\), and the local risk product is zero.
The weights are evaluated only for
\(r\ge\zeta_N>0\); below the guard the estimator returns the grid node, so no
\(0/0\) term is evaluated.
Denote the minimizer in \eqref{supp:eq:sfinite} by
\(\kappa_N^*(r,c_N)\). Define
\begin{equation}
 w^{\rm fin}_{N,qj}(r,c_N)=
 \frac{(1-\eps_N)r\{b_{qj}r^2+2a_{qj}(c_{N,qj}+\kappa_N^*)\}}
 {A^{\rm fin}_{N,qj}(r)}.\label{supp:eq:weights}
\end{equation}
The finite normal equation gives, identically for every \(r\),
\begin{equation}
 \boxed{\sum_{j\ne q}w^{\rm fin}_{N,qj}(r,c_N)a_{qj}=0.}\label{supp:eq:orth}
\end{equation}
This cancels the leading radial mean and the second-order centre nuisance even
when \(r\) is replaced by \(\widehat r\). It also removes any need for a
rate of \(\lambda_N\to\lambda\).

For \(r\ge\zeta_N\), uniformly in the exact finite floor,
\begin{align}
 |w_N|&\le C(R_N+\zeta_N^{-1}),\nonumber\\
 |\partial_rw_N|&\le C(1+\zeta_N^{-2}),&
 |\partial_cw_N|&\le C\zeta_N^{-1}.\label{supp:eq:wbound}
\end{align}
On the pilot event, \(|\widehat r-r|\le d_N/(r+\widehat r)\), so the combined
profile and weight perturbation is at most
\(Cd_N\zeta_N^{-3}=o(1)\). A positive limiting floor makes all denominators
regular; at \(r=0\), viewed as a function of \(r^2\), the required bound is
one-half H\"older and the plug-in error is \(O(d_N^{1/2})\).

\subsubsection{Finite-profile Gaussian approximation and risk transfer}

Let \(n_N=N-m_N\), let \(C_j\) be the independent main-sample counts, and
put
\begin{equation}
 T_N=N^{-1/4}\sum_{j\ne\widehat q}
 w^{\rm fin}_{N,\widehat qj}(\widehat r,c_N)
 \{C_j-n_N\eps_N/M\}.\label{supp:eq:T}
\end{equation}
When \(\widehat r<\zeta_N\), set \(T_N=0\). For a true
local point \((q,r,S)\), define the good pilot event
\begin{equation}
 \mathcal E_N(q,r)=\{\widehat q=q,
 |\widehat r^2-r^2|\le d_N\}.\label{supp:eq:goodpilot}
\end{equation}
Equations \eqref{supp:eq:qpilot}--\eqref{supp:eq:rpilot} imply
\begin{equation}
 \sup_{\substack{q,c_N\in\mathcal K_B\\0\le r\le3R_N}}
 \sqrt N\Pr\{\mathcal E_N(q,r)^c\}\longrightarrow0,\label{supp:eq:pilotnegligible}
\end{equation}
where \(\mathcal K_B\) is any fixed compact body of routing-minimizing
representatives. This is the event on which every conditional calculation
below is made.

\begin{lemma}[Finite-profile Gaussian approximation and uniform risk
transfer]\label{supp:lem:finiteprofile}
Assume \(\lambda_N=\sqrt N\eps_N\to\lambda\in[0,\infty)\), with no rate
condition. Then:
\begin{enumerate}[label=(\Alph*)]
\item Conditionally on \(\mathcal E_N(q,r)\), uniformly in \(q\),
\(c_N\in\mathcal K_B\), and on the guarded growing window
\(\zeta_N\le r\le3R_N\),
\begin{equation}
 \begin{aligned}
 \mathcal L(T_N\mid\mathcal E_N)=\mathcal N\{&S s_{N,q}(r,c_N),\\
 &s_{N,q}(r,c_N)\}+o_{\rm BL}(1).
 \end{aligned}
 \label{supp:eq:finiteprofileCLT}
\end{equation}
The bounded-Lipschitz remainder is \(o(R_N^{-2})\).

\item Let \(R_{N,q,r}^{\rm fin}(c_N)\) be the \(N^{1/2}\)-normalized
conditional squared risk of the plug-in rule
\(\widehat{rS}=\widehat r\tanh(T_N)\), with the grid returned below the
universal guard. After restoring the negligible pilot complement,
\begin{equation}
 \sup_{\substack{q,c_N\in\mathcal K_B\\0\le r\le3R_N}}
 \left|R_{N,q,r}^{\rm fin}(c_N)-
 r^2\psi\{\mu_{q,\lambda}^2(r,[c_N])\}\right|\to0.
 \label{supp:eq:uniformrisktransfer}
\end{equation}
\end{enumerate}
\end{lemma}

\begin{proof}
Conditionally on the pilot, \eqref{supp:eq:T} is a sum over independent frames,
not over independent multinomial cells. A centered one-frame summand is
bounded by
\begin{equation}
 C N^{-1/4}(R_N+\zeta_N^{-1}),\label{supp:eq:summandbound}
\end{equation}
uniformly over the guarded set. Its accumulated third absolute moment is
\begin{equation}
 \begin{aligned}
 C N^{-1/4}\sum_{j\ne q}A^{\rm fin}_{N,qj}|w_{N,qj}|^3
 &\le C N^{-1/4}(R_N^5+\zeta_N^{-1})\\
 &=o(R_N^{-2}).
 \end{aligned}\label{supp:eq:BE}
\end{equation}
using \(A|w|^3\le C(r^5+r^{-1})\), \eqref{supp:eq:rates}, and the absence of the
singular term at a positive limiting floor. The exact conditional variance is
\begin{equation}
 \sum_{j\ne q}A^{\rm fin}_{N,qj}
 \{w^{\rm fin}_{N,qj}(r,c_N)\}^2
 =s_{N,q}(r,c_N),\label{supp:eq:finitevariance}
\end{equation}
up to the harmless factors \(n_N/N\to1\) and the uniformly negligible
multinomial cross-covariance. The signed mean equals \(S s_{N,q}\) to the
same order. Equations \eqref{supp:eq:orth} and \eqref{supp:eq:wbound} make the plug-in
perturbation \(O(d_N\zeta_N^{-3})=o(R_N^{-2})\); the growing-window remainder
is controlled by \eqref{supp:eq:growrem}. Let \(e_N\) be the maximum of the
third-moment, mean, variance, plug-in, and growing-window error bounds just
displayed. The chosen sequences give \(e_N=o(R_N^{-k})\) for every fixed
\(k\). Set \(\rho_N=e_N^{1/2}\). For strengths at least \(\rho_N\),
variance-normalized Berry--Esseen smoothing and rescaling give
bounded-Lipschitz error \(O(e_N/\rho_N)=O(e_N^{1/2})\). Below \(\rho_N\),
Chebyshev's inequality for the centered statistic and its Gaussian target
gives \(O(\rho_N^{1/2}+e_N^{1/2})=O(e_N^{1/4})\). Both are
\(o(R_N^{-2})\), proving \eqref{supp:eq:finiteprofileCLT} without an unstated
positive-variance assumption.

On the guarded window the squared-loss function
\(f_{r,\widehat r,S}(t)=(\widehat r\tanh t-rS)^2\) has supremum and Lipschitz
constant \(O(R_N^2)\). Applying Part A to
\(f_{r,\widehat r,S}/(CR_N^2)\), and then using the pilot perturbation bounds,
therefore gives the finite-profile risk
\begin{equation}
 R_{N,q,r}^{\rm fin}(c_N)=r^2\psi\{s_{N,q}(r,c_N)\}+o(1)
 \label{supp:eq:finiteprofilerisk}
\end{equation}
uniformly on the guard; \eqref{supp:eq:pilotnegligible} removes conditioning. It
remains to transfer the risk, not the strength, to the limiting profile.

Fix \(0<\delta<R<\infty\) and split
\([0,3R_N]=[0,\delta]\cup[\delta,R]\cup[R,3R_N]\).
For the small zone, \(0\le\psi\le1\) gives exactly
\begin{equation}
 \sup_{r\le\delta}\left|r^2\psi(s_{N,q})-
 r^2\psi(\mu_{q,\lambda}^2)\right|\le2\delta^2.
 \label{supp:eq:smallrisk}
\end{equation}
This also covers the zero-critical-floor guard and is the step that absorbs an
arbitrarily slow \(\lambda_N\to0\).

On \([\delta,R]\), every denominator is uniformly bounded away from its
singular endpoint. Since \(\lambda_N\to\lambda\), \(\eps_N\to0\), and
\(\mathcal K_B\) is compact, continuity of a strictly convex one-dimensional
quadratic profile gives
\begin{equation}
 \sup_{\substack{q,c_N\in\mathcal K_B\\\delta\le r\le R}}
 |s_{N,q}(r,c_N)-\mu_{q,\lambda}^2(r,[c_N])|\to0.
 \label{supp:eq:compactstrength}
\end{equation}
The strengths occupy a compact interval, so uniform continuity of \(\psi\)
transfers the intermediate-zone risk.

For the large zone, pair \(j=q+d\) with \(j=q-d\). Their \(b\)'s cancel,
so the finite and limiting profile normal equations give uniformly bounded
centres on \(\mathcal K_B\). Consequently, uniformly for all sufficiently
large fixed \(R\) and all large \(N\),
\begin{equation}
 \begin{gathered}
 s_{N,q}(r,c_N)\ge cr^4,\qquad
 \mu_{q,\lambda}^2(r,[c_N])\ge cr^4,\\
 R\le r\le3R_N.
 \end{gathered}
 \label{supp:eq:cubiccoercivity}
\end{equation}
Indeed, after division by \(r^4\), both leading constants converge uniformly
to \(\sum_{j\ne q}b_{qj}^2/a_{qj}>0\); the fixed channel \(j=q+1\) proves
strict positivity, including even \(M\). The hard-sign rule bounds the
binary Gaussian Bayes risk by
\(\psi(s)\le4\Phi(-\sqrt s)\le2e^{-s/2}\). Hence
\begin{equation}
 r^2\psi(s_{N,q})+r^2\psi(\mu_{q,\lambda}^2)
 \le Cr^2e^{-cr^4},\label{supp:eq:largerisk}
\end{equation}
whose supremum over \(r\ge R\) vanishes as \(R\to\infty\).

The order of limits is essential: fix \((\delta,R)\), let \(N\to\infty\),
then let \(\delta\downarrow0\), and finally let \(R\uparrow\infty\).
Equations \eqref{supp:eq:smallrisk}, \eqref{supp:eq:compactstrength}, and
\eqref{supp:eq:largerisk} prove Part B. Uniform convergence of the finite
sign strength fails at \(r=0\), so the argument transfers the risk-weighted
quantity directly; the prefactor \(r^2\) suppresses the nonuniform regime.
\end{proof}

A tempting stronger uniform-strength assertion is false.  For \(M=4\), put
zero on cycle edges and one on diameters, take
\(\lambda_N=1/\log N\) and \(r_N=2N^{-1/32}\).  Then
\(\mu_{0,0}^2(r_N,[c])\to4/5\) but \(s_{N,0}(r_N,c)\to0\), while the
risk-weighted discrepancy is at most \(2r_N^2\to0\). This counterexample
explains why the theorem asserts uniform risk transfer rather than uniform
convergence of the raw sign strength.

\subsubsection{Uniform outer Hellinger inverse}

For every node select \(j(q)=q+1\pmod M\). Put \(t_N=\etaN^2\). Analyticity
of the exact near-DFT probability gives its odd part
\begin{equation}
 \begin{aligned}
 p_N(h)-p_N(-h)={}&2b_{qj}h^3+4a_{qj}c_{N,qj}t_Nh\\
 &+O(h^5+t_Nh^3+t_N^2h).
 \end{aligned}\label{supp:eq:odd}
\end{equation}
When \(R_N\etaN\le|h|\le h_0\), the perturbation-to-cubic ratios are at most
\(R_N^{-2}\), \(\etaN^2R_N^{-2}\), and \(h_0^2\). Choose \(h_0\) small and
then \(N\) large. Since \(b_{q,q+1}\ne0\),
\(|p_N(h)-p_N(-h)|\ge c|h|^3\). The sum of the two replaced probabilities is
at most \(Ch^2\), because \(\eps_N=O(\etaN^2)=O(h^2/R_N^2)\). Thus
\begin{equation}
 H^2(\widetilde p_{\theta_q+h,N}^{U_N},
     \widetilde p_{\theta_q-h,N}^{U_N})\ge c_1h^4.\label{supp:eq:oppH}
\end{equation}

For same-side points whose segment stays outside the local band, the selected
square-root coordinate has derivative
\[
 \frac{(1-\eps_N)p_N'(h)}{2\sqrt{(1-\eps_N)p_N(h)+\eps_N/M}}.
\]
Here \(p_N'(h)=2a_{qj}h+O(h^2+t_N)\), its sign is fixed, and the denominator
is \(O(|h|)\). Integration gives
\begin{equation}
 H^2(\widetilde p_{h_1,N}^{U_N},\widetilde p_{h_2,N}^{U_N})
 \ge c_2|h_1-h_2|^2.\label{supp:eq:sameH}
\end{equation}
For \(h,k\in[R_N\etaN,h_0]\), opposite unequal radii obey
\begin{equation}
 H^2(\widetilde p_{\theta_q+h,N}^{U_N},
     \widetilde p_{\theta_q-k,N}^{U_N})\ge
 c\{(h-k)^2+\min(h,k)^4\}.\label{supp:eq:unequalH}
\end{equation}
Use the paired cells \(j=q\pm1\), which have the same \(a>0\) and opposite
nonzero \(b\). For \(m=\min(h,k)\),
\(\eps_N/m^2,t_N/m^2=O(R_N^{-2})\), and their square-root differences satisfy
\[
\begin{aligned}
 \Delta_{\pm}&=\sqrt a\,(h-k)
 \mathbin{\pm}\frac{b}{2\sqrt a}(h^2+k^2)+e_{\pm,N},\\
 |e_{+,N}|+|e_{-,N}|&\le C(h_0+R_N^{-2})\\
 &\quad{}\times\{|h-k|+h^2+k^2\}.
\end{aligned}
\]
The leading two-by-two map is uniformly nonsingular. Choose \(h_0\) small and
then \(N\) large; its squared norm, with the cross term canceled, dominates
\((h-k)^2+(h^2+k^2)^2\). The error is a small perturbation, and the full-law
distance dominates these two coordinates, proving \eqref{supp:eq:unequalH}
uniformly.
Outside fixed grid neighborhoods, exact-router injectivity first gives a
uniform nonzero Hellinger separation for pairs at circular distance at least
any fixed \(\rho>0\). A local metric bound on compact sets separated from
the routing nodes follows from the exact-router moment.
Differentiating the exact first moment displayed in S1 gives
\begin{equation}
 m_1'(\theta)=\frac{i(M-1)}M e^{i\theta}(1-e^{-iM\theta}),\label{supp:eq:m1prime}
\end{equation}
which is nonzero away from the grid. Every exact-router coordinate probability
is positive there. Hence the derivative of
\(\theta\mapsto(\sqrt{p_j(\theta)})_{j=0}^{M-1}\) is nonzero; on a fixed
compact set its norm has a positive minimum. Uniform Taylor expansion gives
constants \(c,\rho>0\) such that
\begin{equation}
 H^2(P_\theta,P_{\theta'})\ge
 c\distT(\theta,\theta')^2,\qquad
 \distT(\theta,\theta')\le\rho.\label{supp:eq:compactHmetric}
\end{equation}
The \(C^1\) convergence of the finite near-router laws transfers this bound,
and compact injectivity transfers the nonlocal separation, uniformly for all
large \(N\).

Let \(\widehat p\) be the main empirical law and let
\(\Theta_N^{\rm out}\) be the circle with the open \(R_N\etaN\) grid
neighborhoods removed. Define, with smallest-angle tie breaking,
\begin{equation}
 \widehat\theta_{\rm out}\in\argmin_{\phi\in\Theta_N^{\rm out}}
 H(\widehat p,\widetilde p_{\phi,N}^{U_N}).\label{supp:eq:MHD}
\end{equation}
Compactness gives existence and measurability. For \(n\) multinomial trials,
\[
 \E_pH^2(\widehat p,p)\le M/n,
\]
because \((\sqrt{\widehat p_j}-\sqrt p_j)^2\le
(\widehat p_j-p_j)^2/p_j\) whenever \(p_j>0\). Moreover the method of types
and \(D(q\|p)\ge H^2(q,p)\) give
\begin{equation}
 \Pr_p\{H^2(\widehat p,p)\ge u\}\le(n+1)^Me^{-nu}.\label{supp:eq:types}
\end{equation}
Equations \eqref{supp:eq:sameH}--\eqref{supp:eq:types} yield ordinary \(O(N^{-1})\)
correct-branch risk, while \eqref{supp:eq:oppH} gives side-confusion probability
at most \((N+1)^Me^{-cNh^4}\). Hence
\begin{equation}
 \sup_{\theta\in\Theta_N^{\rm out}}
 \E\distT^2(\widehat\theta_{\rm out},\theta)=O(N^{-1}).\label{supp:eq:outerrisk}
\end{equation}

\subsubsection{Composite measurable estimator}

Invoke the local rule when the pilot identifies a routed node and
\(\widehat r\le2R_N\); invoke \eqref{supp:eq:MHD} otherwise. Partition the true
parameter into \(r\le R_N\), the transition band \(R_N<r<3R_N\), and the
outer region. Wrong-node, false-local, and false-outer contributions vanish
after multiplication by \(\sqrt N\) by
\eqref{supp:eq:qpilot}, \eqref{supp:eq:pilotbranch}, and bounded circular loss. Side
confusion vanishes by \eqref{supp:eq:rates} and \eqref{supp:eq:types}. In the transition
band the local estimates above apply through \(3R_N\), giving risk
\(O(r^2e^{-cr^4})\), while the outer risk is \(O(N^{-1})\); hence either
branch is \(o(N^{-1/2})\). Therefore one estimator
satisfies
\begin{align}
 &\sup_{q,0\le r\le R_N}
 \left|\sqrt N\E\distT^2(\widehat\theta,\theta)
 -r^2\psi(\mu_{q,\lambda}^2(r,[c]))\right|\to0,\label{supp:eq:localuniform}\\
 &\sup_{\distT(\theta,\{\theta_q\})>R_N\etaN}
 \sqrt N\E\distT^2(\widehat\theta,\theta)\to0.\label{supp:eq:outeruniform}
\end{align}

Choose an optimizer of \eqref{supp:eq:C}, its routing-minimizing representative,
and the phase-aligned tangent, shrunk by \(1-o(1)\) if needed. Equations
\eqref{supp:eq:localuniform}--\eqref{supp:eq:outeruniform} give
\begin{equation}
 \limsup_N\sqrt N\,\mathcal R_N^*(\tau,\eps_N)
 \le\cM(\tau,\lambda).\label{supp:eq:upper}
\end{equation}
Combining \eqref{supp:eq:lower} and \eqref{supp:eq:upper} proves the sharp equality.

\section{Quotient geometry}

This section records the exact normal form behind the compatibility
calculations.  For a fixed node and radius, put
\begin{equation}
 w_{qj}=\frac{a_{qj}^{2}}{a_{qj}r^{2}+\lambda/M},\qquad
 u_{qj}=\frac{b_{qj}r^{2}}{2a_{qj}},
 \label{supp:eq:suppwu}
\end{equation}
and define the weighted star quotient seminorm
\[
 \|[z]\|_{q,r,\lambda}^{2}
 =\min_{\alpha\in\R}\sum_{j\ne q}w_{qj}(z_{qj}+\alpha)^2.
\]
Factoring \(2a_{qj}\) from every numerator in the main profile gives,
identically,
\begin{equation}
 \mu_{q,\lambda}^{2}(r,[c])
 =4r^{2}\|[c_q]+[u_q(r)]\|_{q,r,\lambda}^{2}.
 \label{supp:eq:suppquotientdistance}
\end{equation}
The minimizer is the weighted mean
\(
 \alpha^*=-\sum_jw_{qj}(c_{qj}+u_{qj})/\sum_jw_{qj}
\), so it is unique for \(r>0\).  This proves both quotient invariance and
the claim that the design approximates all radius curves with one shared
edge field.  Differentiating at \(c=0\), where antisymmetry makes
\(\alpha^*=0\), gives the divergence operator in the main paper.

For a line \(c=tg\), profiling a positive quadratic form in \((t,\alpha)\)
produces an exact quadratic polynomial in \(t\).  If the divergence term
vanishes, completing the square leaves
\begin{equation}
 \mu_{q,\lambda}^{2}(r,tg)-\mu_{q,\lambda}^{2}(r,0)
 =4t^{2}r^{2}\min_{\alpha}
 \sum_{j\ne q}w_{qj}(g_{qj}+\alpha)^2.
 \label{supp:eq:suppcurvature}
\end{equation}
All weights are positive.  Equality therefore means exactly that the star is
constant. Thus no cubic or higher-order term occurs along a line in the
canonical profile.

\section{Positive-Critical-Replacement Compatibility Spectrum}

Let \(L=\lfloor(M-1)/2\rfloor\), \(\lambda>0\), and use unordered distance
coordinates \(y_d(q)=g_{\{q,q+d\}}\), \(1\le d\le L\).  The restriction
\(d<M/2\) makes this representation unique.  If \(M\) is even, distance
\(M/2\) is instead a matching of \(M/2\) unordered edges; its coefficient
\(b\) is zero, so these variables are free at every radius.

Set \(a_d=\{4\sin^2(\pi d/M)\}^{-1}\),
\(\beta_d=\cos(\pi d/M)/\{4\sin^3(\pi d/M)\}\), and
\(\rho_d=\lambda/(Ma_d)\).
For \(1\le d\le L\), all \(\beta_d\) are nonzero and the \(a_d\), hence the
\(\rho_d\), are pairwise distinct.  With orientation \(q\to q+d\), the
neutrality equation at \(x_s=r_s^2>0\), after deletion of a common factor,
is
\begin{equation}
 \sum_{d=1}^{L}\frac{\beta_d}{x_s+\rho_d}
 \{y_d(q)-y_d(q-d)\}=0.
 \label{supp:eq:suppbalance}
\end{equation}

Complexify this real system and apply the length-\(M\) DFT in \(q\).  In
character \(\ell\), the \(d\)th difference acquires the factor
\(1-e^{-2\pi i\ell d/M}\), which vanishes precisely when \(M\mid\ell d\).
Thus
\[
 \iota_\ell=\#\{1\le d\le L:M\mid\ell d\}
 =\left\lfloor\frac{L\gcd(M,\ell)}M\right\rfloor
\]
columns are inactive. On the remaining columns the \(k\) radius evaluations give
a Cauchy matrix, up to invertible column scalings.  Every \(m\times m\) minor
has determinant
\begin{equation}
 \begin{aligned}
 \det\!\left[\frac1{x_{s_u}+\rho_{d_v}}\right]_{u,v=1}^{m}
 &={} \frac{\prod_{u<v}(x_{s_v}-x_{s_u})}
 {\prod_{u,v}(x_{s_u}+\rho_{d_v})}\\
 &\quad{}\times\prod_{u<v}(\rho_{d_v}-\rho_{d_u}),
 \end{aligned}
 \label{supp:eq:suppcauchy}
\end{equation}
up to a fixed sign, and is nonzero for distinct radii.  The block rank is
therefore \(\min\{k,L-\iota_\ell\}\).

The stacked map is real, so its real and complex ranks coincide; the complex
DFT is invertible.  Thus its real kernel dimension sums all \(M\) complex
block nullities, not one per conjugate pair.  Adding the \(M/2\) diameter
variables when even and removing the global constant-edge gauge gives
\begin{equation}
 \dim_{\R}\NM^{(k)}=
 \sum_{\ell=0}^{M-1}\left[L-\min\{k,L-\iota_\ell\}\right]
 +\ind_{\{2\mid M\}}\frac M2-1.
 \label{supp:eq:suppspectrum}
\end{equation}

At \(k=L\), every active block has full column rank.  The remaining inactive
coefficients are exactly \(y_d(q)=y_d(q-d)\), a condition sufficient at every
radius and constant on the \(\gcd(M,d)\) cosets generated by \(d\).  Hence
\begin{equation}
 \dim\NM(\lambda)=\sum_{d=1}^{L}\gcd(M,d)
 +\ind_{\{2\mid M\}}\frac M2-1.
 \label{supp:eq:suppterminaldim}
\end{equation}
Sharpness is radius-choice independent. If \(k<L\), character \(\ell=1\)
has no inactive columns and Cauchy nullity \(L-k\); conjugate pairing turns a
nonzero null vector into a real field neutral at the evaluated radii but not
persistent.  Repeats add no rank, and \(r=0\) gives the zero operator.  At
\(\lambda=0\), all poles merge, positive-radius rows are proportional, and
one radius certifies dimension \(M(M-3)/2\).

Exact rank therefore persists for every \(\lambda>0\), but it does not
provide a uniform stability guarantee as magnitudes coalesce or as the
positive critical replacement floor approaches zero.

At positive floor, a persistent star assigns one value per undirected
separation, so quotienting its constant value gives curvature rank
\(\lfloor M/2\rfloor-1\).  At zero floor, rescale active oriented edges by
\(b_{qj}\).  Values at \(q\) obey one divergence constraint; every compatible
remaining imbalance is routable on a spanning tree after deleting \(q\).
With a free even diameter this attains star dimension \(M-2\), hence rank
\(M-3\) modulo constants.  Finally, if every star is constant, shared edges
make all constants equal, proving positive definiteness of summed curvature
on the global quotient.

\section{Global three-channel rigidity}

For \(M=3\), \(a_{q,q\pm1}=1/3\) and
\(b_{q,q\pm1}=\pm1/(3\sqrt3)\).  Let
\(\delta_q=c_{q,q+1}-c_{q,q-1}\), \(x=r^2\), and
\begin{equation}
\begin{aligned}
 f_\lambda(\delta,x)&=x\psi\!\left[
 \frac{2x}{3(x+\lambda)}\left(\delta+\frac{x}{\sqrt3}\right)^2\right],\\
 G_\lambda(\delta)&=\sup_{x\ge0}f_\lambda(\delta,x).
\end{aligned}
\label{supp:eq:suppM3profile}
\end{equation}
The extension at \(x=0\) is zero.  Since \(\psi(s)\) decays exponentially,
the supremum is positive and attained in the interior; over every compact
\(\delta\)-set, all maximizers lie in one compact subinterval of
\((0,\infty)\).  If \(\delta<0\), put \(x_0=-\sqrt3\delta\).  For \(x<x_0\),
\(f_\lambda(\delta,x)\le x<x_0\); at \(x_0\) it equals \(x_0\); and at
\(x_0+\epsilon\) it is \(x_0+\epsilon+O(\epsilon^2)>x_0\).  Thus every maximizer satisfies
\(\delta+x/\sqrt3>0\).

At every active \(x\),
\[
 \partial_\delta f_\lambda(\delta,x)=
 \frac{4x^2}{3(x+\lambda)}
 \left(\delta+\frac{x}{\sqrt3}\right)\psi'(s_\lambda)<0.
\]
For \(h_n\downarrow0\), compactness gives a subsequential limit of maximizers
of \(f_\lambda(\delta+h_n,\cdot)\) that maximizes at \(\delta\). The
mean-value theorem bounds the upper right Dini derivative of \(G_\lambda\) by
the largest displayed derivative on this compact active set, hence strictly
below zero. The Dini mean-value lemma therefore makes \(G_\lambda\) strictly
decreasing.

Finally, \(\delta_0+\delta_1+\delta_2=0\), and all three vanish only for a
constant edge field.  Any nonzero quotient has a negative \(\delta_q\), so
\(F_{3,\lambda}([c])\ge G_\lambda(\delta_q)>G_\lambda(0)\).  The zero class is
therefore the unique global canonical minimizer for every finite resource.
For \(M\ge4\), the nearest-neighbor ring field is persistent at every floor
and radius and has a nonconstant star at every node; the exact curvature
identity gives the strict opposite conclusion.

\section{Arithmetic Localization of Curvature Visibility}

At a positive critical replacement floor, call \(q\) curvature-visible when
its star is nonconstant, equivalently when its profiled star curvature is
positive; this is radius- and gauge-independent. If \(S([g])\) is not all
nodes, gauge one invisible star to zero. Shared edges then force every other
invisible star to zero, so each nonzero edge lies inside \(S([g])\). A nonzero
persistent distance-\(d\) edge propagates around
an orbit of \(M/\gcd(M,d)\ge p_{\min}(M)\) nodes; the all-visible case gives
the same bound.  An even diameter attains two.  For odd \(M\ge5\), put value
one on one complete orbit with \(d=M/p_{\min}(M)\) and zero elsewhere; exactly
that orbit is visible.  Thus the minimum support is \(p_{\min}(M)\), and every
nonzero odd-\(M\) mode is globally visible iff \(M\) is prime.  The \(M=3\)
exclusion is essential because its three-cycle is the constant quotient.

At zero floor set \(h_{qj}=b_{qj}g_{qj}\), so persistence is divergence
freedom. The all-visible odd bound is immediate; otherwise the preceding gauge
confines every nonzero edge to \(S([g])\). For odd \(M\), \(h\) and \(g\) have
the same support, and a nonzero divergence-free \(h\) contains a cycle, giving
\(|S([g])|\ge3\); a triangle circulation pulled back by \(g=h/b\) attains three
for odd \(M\ge5\). For even \(M\), the same support argument gives two, attained
by one zero-weight diameter.

\section{Critical Replacement-Floor Geometry and Value Continuity}

Let \(P_0\) and \(P_+\) project orthogonally onto \(\NM(0)\) and the
positive-floor space. Persistence gives \(\NM(\lambda)\subseteq\NM(0)\), with
equality for \(M=3,4\) and proper inclusion for \(M\ge5\) by the dimension
formulas. Hence \(\|P_0-P_+\|_{\rm op}\) is respectively zero or one; in the
proper case, use a unit vector in \(\NM(0)\cap\NM(\lambda)^\perp\).

On bounded design balls and compact \(\lambda\)-sets, split radii into
\(r\le\delta\),
\([\delta,R]\), and \(r\ge R\). The respective controls are risk at most
\(\delta^2\), compact continuity of the profiled quadratic, and, by paired
antisymmetry, \(\mu^2\ge cr^4\) with risk at most \(Cr^2e^{-cr^4}\). Taking
parameter limits before \(\delta\downarrow0\) and \(R\uparrow\infty\) proves
uniform continuity of \(F\) on bounded design balls. Since
\([c]=\sqrt\tau[d]\), \(\|[d]\|_{\rm rt}\le1\), places all minima on one
compact ball, compact minimization proves joint continuity of
\(\cM(\tau,\lambda)\). Feasible-set nesting and denominator monotonicity give
weak decrease in \(\tau\) and weak increase in \(\lambda\).

\bibliographystyle{IEEEtran}
\bibliography{references}
\end{document}